\documentclass[11pt,letterpaper]{article} 
\usepackage[margin=1in]{geometry}%
\usepackage{amsmath,amssymb,amsthm}
\usepackage{graphicx}
\usepackage[linktocpage=true,
pagebackref=true,colorlinks,
linkcolor=DarkRed,citecolor=ForestGreen,
bookmarks,bookmarksopen,bookmarksnumbered]
{hyperref}
\usepackage[capitalise]{cleveref}

\usepackage{framed}
\usepackage{mdframed}

\definecolor{DarkGray}{rgb}{0.66, 0.66, 0.66}
\definecolor{DarkPowderBlue}{rgb}{0.0, 0.2, 0.6}
\definecolor{fluorescentyellow}{rgb}{0.8, 1.0, 0.0}
\definecolor{ForestGreen}{rgb}{0.1333,0.5451,0.1333}
\definecolor{DarkRed}{rgb}{0.65,0,0}

\Crefname{algocf}{Algorithm}{Algorithms}
\crefname{algocfline}{line}{lines}
\Crefname{invariant}{Invariant}{Invariants}
\Crefname{claims}{Claim}{Claims}

\usepackage{mathrsfs}
\usepackage{xspace}
\usepackage{soul}
\usepackage{latexsym}
\usepackage{bbm}
\usepackage{thm-restate}
\usepackage{nicefrac}            %
\usepackage{microtype}           %

\usepackage{thmtools,thm-restate}

\ifdefined\DEBUG
\newcommand{\alert}[1]{{\color{red}#1}}
\newcommand{\agnote}[1]{\refstepcounter{note}$\ll${\bf Anupam~\thenote:}
  {\sf \color{blue} #1}$\gg$\marginpar{\tiny\bf AG~\thenote}}
\newcommand{\elnote}[1]{\refstepcounter{note}$\ll${\bf E~\thenote:}
  {\sf \color{gray} #1}$\gg$\marginpar{\tiny\bf EL~\thenote}}
\newcommand{\spnote}[1]{\refstepcounter{note}$\ll${\bf S~\thenote:}
  {\sf \color{magenta} #1}$\gg$\marginpar{\tiny\bf SP~\thenote}}
\newcommand{\dpnote}[1]{\refstepcounter{note}$\ll${\bf Debmalya~\thenote:}
  {\sf \color{magenta} #1}$\gg$\marginpar{\tiny\bf DP~\thenote}}  
\newcommand{\tdnote}[1]
{\refstepcounter{note}$\ll${\bf Tommaso~\thenote:}
  {\sf \color{magenta} #1}$\gg$\marginpar{\tiny\bf TD~\thenote}} 
  
\else
\newcommand{\alert}[1]{}
\newcommand{\agnote}[1]{}
\newcommand{\elnote}[1]{}
\newcommand{\dpnote}[1]{}
\newcommand{\tdnote}[1]{}
\newcommand{\spnote}[1]{}
 \fi

\newcommand{\initOneLiners}{%
    \setlength{\itemsep}{0pt}
    \setlength{\parsep }{0pt}
    \setlength{\topsep }{0pt}
}

\newtheorem{theorem}{Theorem}[section]
\newtheorem{lemma}[theorem]{Lemma}
\newtheorem{claims}[theorem]{Claim}

\newtheorem{corollary}[theorem]{Corollary}
\newtheorem{remark}[theorem]{Remark}
\newtheorem{definition}[theorem]{Definition}

\newcommand{\eat}[1]{}

\newcommand{\sse}{\subseteq}

\newcommand{\calC}{\mathcal{C}}

\newcommand{\calF}{\mathcal{F}}
\newcommand{\calP}{\mathcal{P}}

\newcommand{\calS}{\mathcal{S}}

\newcommand{\poly}{\operatorname{poly}}

\renewcommand{\emptyset}{\varnothing}

\newcommand{\cost}{cost}
\newcommand{\shortbar}{\text{\fontfamily{times}\selectfont-}}
\newcommand{\SymRel}{{\rm SymRel}}
\newcommand{\SymLang}{{\rm SymLang\shortbar N}}
\newcommand{\AND}{{\rm AND}}
\newcommand{\SAT}{{\rm SAT}}
\newcommand{\ALLEQ}{{\rm AE}}
\newcommand{\LE}{\le1{\rm \shortbar out \shortbar of \shortbar}r{\rm SAT}}
\newcommand{\PMC}{\text{\sc Paired\ Minimum}\ st\text{\sc \shortbar Cut}}

\newcommand{\CSP}{%
  \ifmmode
    \operatorname{\textsc{CSP}}%
  \else
    \textsc{CSP}%
  \fi
}

\newcommand{\MaxCSP}{\textsc{MaxCSP}}

\newcommand{\MinCSP}{%
  \ifmmode
    \operatorname{\textsc{MinCSP}}%
  \else
    \textsc{MinCSP}%
  \fi
}

\newcommand{\MinCSPwP}{%
  \ifmmode
    \operatorname{
    \textsc{ImproveMaxCSP}
    }%
  \else
    \textsc{ImproveMaxCSP}%
  \fi
}

\newcommand{\CUT}{%
  \ifmmode
    \operatorname{%
      \textsc{ImproveMaxCut}\allowbreak\shortbar\allowbreak\textsc{Uncut}
    }%
  \else
    \textsc{ImproveMaxCut\shortbar{}Uncut}%
  \fi\xspace
}

\newcommand{\CUTTERMINAL}{%
  \ifmmode
    \operatorname{%
      \textsc{ImproveMaxCut}\allowbreak-\allowbreak\textsc{Uncut}\allowbreak\textsc{Terminal}\allowbreak\textsc{Version}%
    }%
  \else
    \textsc{ImproveMaxCut\shortbar{}Uncut Terminal Version}%
  \fi\xspace
}

\newcommand{\MISVW}{%
  \ifmmode
    \operatorname{%
      \textsc{Maximum}\allowbreak\ \textsc{Induced}\allowbreak\ \textsc{Subhypergraph}\allowbreak\ \textsc{with}\allowbreak\ \textsc{Vertex}\allowbreak\ \textsc{Weights}%
    }%
  \else
    \textsc{Maximum Induced Subhypergraph with Vertex Weights}%
  \fi\xspace
}

\newcommand{\MIS}{%
  \ifmmode
    \operatorname{%
      \textsc{Multicolored}\allowbreak\ \textsc{Independent}\allowbreak\ \textsc{Set}%
    }%
  \else
    \textsc{Multicolored Independent Set}%
  \fi\xspace
}

\title{Complexity of Local Search for CSPs Parameterized by Constraint Difference}

\author{
  Aditya Anand\thanks{University of Michigan, Ann Arbor, USA. \texttt{adanand@umich.edu}.}
  \and Vincent Cohen-Addad\thanks{Google Research, NYC, USA. \texttt{vcohenad@gmail.com}.}
  \and Tommaso D'Orsi\thanks{Bocconi University, Italy. \texttt{tomasso.dorsi@unibocconi.it}.}
  \and Anupam Gupta\thanks{New York University, USA. \texttt{anupam.g@nyu.edu}. Supported in part by NSF awards CCF-2224718 and CCF-2422926.}
  \and Euiwoong Lee\thanks{University of Michigan, Ann Arbor, USA. \texttt{euiwoong@umich.edu}. Supported in part by NSF award CCF-2236669 and by Google, Inc.}
  \and Debmalya Panigrahi\thanks{Duke University, USA. \texttt{debmalya@cs.duke.edu}. Supported in part by NSF awards CCF-1955703 and CCF-2329230.}
  \and Sijin Peng\thanks{Massachusetts Institute of Technology, USA. \texttt{peng-sj21@mails.tsinghua.edu.cn}.}
}

\date{}  %

\begin{document}

\date{}
\maketitle

\begin{abstract}
    In this paper, we study the parameterized complexity of local search, whose goal is to find a good nearby solution from the given current solution. Formally, given an optimization problem where the goal is to find the largest {\em feasible} subset $S$ of a universe $U$, the new input consists of a {\em current solution} $P$ (not necessarily feasible) as well as an ordinary input for the problem. 
    Given the existence of a feasible solution $S^*$, the goal is to find a feasible solution as good as $S^*$ in parameterized time $f(k) \cdot n^{O(1)}$, where $k$ denotes the {\em distance} $|P\Delta S^*|$. This model generalizes numerous classical parameterized optimization problems whose parameter $k$ is the minimum number of elements removed from $U$ to make it feasible, which corresponds to the case $P = U$. 

    We apply this model to widely studied Constraint Satisfaction Problems (CSPs), where $U$ is the set of constraints, and a subset $U'$ of constraints is feasible if there is an assignment to the variables satisfying all constraints in $U'$. We give a complete characterization of the parameterized complexity of all boolean-alphabet {\em symmetric} CSPs, where the predicate's acceptance depends on the number of true literals. 
    
\end{abstract}
\newpage
\tableofcontents
\newpage

\section{Introduction}\label{sec:intro}The classical theory of NP-hardness sets limits on the polynomial-time tractability of a vast array of algorithmic problems. This raises the question: {\em what additional information about the problem instance can one use to help overcome complexity barriers?} A very successful line of research that aims to address this broad question is that of {\em fixed parameter tractable} (FPT) algorithms. The idea is to identify a parameter that has a small value in {\em typical} instances, and design an algorithm whose running time is exponential in the value of the parameter but polynomial in the overall size of the instance. 

This philosophy has been widely applied to the following general class of optimization problems: let $U$ be a universe and let $\calS \subseteq 2^U$ be an (often implicitly given) collection of {\em feasible} subsets of $U$, and the goal is to find $S \in \calS$ that maximizes $|S|$, or equivalently to find the smallest $T$ that puts $U \setminus T \in \calS$. For instance, given a graph $G = (V, E)$, if $U = E$ and $F \in \calS$ if $(V, F)$ is $2$-colorable, the problem corresponds to the famous \textsc{MaxCut}/\textsc{MinUncut} pair. Or given $G = (V, E)$, if $U = V$ and $\calS$ denotes the collection of all independent sets, then we have \textsc{Independent Set}/\textsc{Vertex Cover}. 
In this class, there are many situations where typical (or at least interesting) instances have the optimal value for the minimization version much smaller than $|U|$. (I.e., the universe $U$ is already close to feasible.) Consequently, numerous results in the FPT literature study a problem from the above class - we let the parameter $k$ be the optimal value for the minimization version, and give an algorithm of running time $f(k)n^{O(1)}$ (better than the naive running time of $n^{O(k)}$).

Can we extend the applicability of this framework? 
We can interpret the algorithms in the above paragraph as one-step {\em local search}, where we are given a current solution (not necessarily feasible) {\em close} to some feasible solution $S^*$ with the optimal value, and the algorithm finds a feasible solution as good as $S^*$.
This motivates us to 
ask: {\em if the algorithm is provided {\bf any} current solution $P \subseteq U$ (as a solution for the max version) and there is a feasible solution $S^* \subseteq U$ close to $P$, can the algorithm efficiently recover a solution as good as $S^*$?}  Formally, we define the parameter $k$ as the distance $|P \Delta S^*|$, and our goal is to design an algorithm whose runtime is arbitrary in $k$ but polynomial in the size of the problem instance.\footnote{Actually, since $|S \Delta S^*| = |(U\setminus S) \Delta (U \setminus S^*)|$, this parameter remains the same for both maximization and minimization versions.} %

We apply this model to Constraint Satisfaction Problems (CSPs), one of the most important classes of optimization problems in the theory of computation. 
Given a predicate $R \subseteq \{ 0, 1 \}^r$, \MaxCSP($R$) is the computational problem whose input consists of the set of variables $V$ and the set of constraints $\mathcal{C}$. 
Given an assignment $\alpha : V \to \{ 0, 1 \}$, each constraint $C \in \mathcal{C}$ is satisfied when the values of the $r$ specified literals (variables or their negations) belong to $R$. 
(See Section~\ref{sec:preliminaries} for formal definitions.)
The goal is to find an assignment that satisfies as many constraints as possible. 
Generalizing the \textsc{MaxCut} example above, we apply \MaxCSP($R$) to our model where $U = \calC$ and a subset of constraints $\calC' \subseteq \calC$ is feasible if there exists an assignment $\alpha$ that satisfies all of $\calC'$ (i.e., $\calC'$ is {\em satisfiable}); so the current solution $\calP$ is a set of constraints and we assume that it is close to some set of satisfiable constraints. 
It yields the following computational task.

\begin{framed}
\noindent \textbf{Problem:} \MinCSPwP$(R)$

\noindent\textbf{Input:} A $\CSP(R)$ instance $I(V,\mathcal{C})$, an integer $k \in \mathbb{N}$, and $\mathcal{P} \subseteq \mathcal{C}$.

\noindent\textbf{Parameter:} $k$

\noindent\textbf{Promise: } There exists an assignment $\alpha_0$ with the property $|\mathcal{C}_I(\alpha_0) \Delta \mathcal{P}| \le k$, where $ \mathcal{C}_I (\alpha_0)$ is the set of clauses $\alpha_0$ satisfies, and $\Delta$ represents symmetric difference.

\noindent\textbf{Output:} A {\em good} assignment $\alpha$, where the word {\em good} simply means $|\mathcal{C}_I(\alpha)| \geq 
|\mathcal{C}_I(\alpha_0)|$.
(We do not require $\alpha$ to satisfy $|\mathcal{C}_I(\alpha) \Delta \mathcal{P}| \le k$.)
\end{framed}

Note that, given an instance $(I(V, \calC)$, $k$, $\calP$), the problem is essentially equivalent to find $\alpha$ such that $|\calC_I(\alpha)|$ is at least $\max_{\beta : |\calC_I(\beta) \Delta \calP| \leq k} |\calC_I(\beta)|$; in words, find a feasible solution as good as any feasible $k$-neighbor of $\calP$. (Formally, we still will keep the fixed promised assignment $\alpha_0$ as part of the instance and define a good assignment based on $\alpha_0$.)

We also observe that if the decision version of $\CSP(R)$ (i.e., finding an assignment that satisfies every constraint) can be solved in polynomial time, an easy $|\calC|^k \cdot \poly(|V|, |\calC|)$-time algorithm for \MinCSPwP$(R)$ follows by exhaustively guessing $(\mathcal{C}_I(\alpha_0) \Delta \mathcal{P})$ and solving the decision version with constraints $\mathcal{C}_I(\alpha_0)$; otherwise it is NP-hard even when $k = 0$ (just letting $\calP = \calC$), which concludes that \MinCSPwP$(R)$ is in $\mathbf{XP}$ if and only if the decision version is in $\mathbf{P}$. Also, if \textsc{MaxCSP}$(R)$ is NP-hard, there is no $\poly(|V|, |\calC|,k)$-time algorithm for \MinCSPwP$(R)$.

In this paper, we provide a complete characterization of the parameterized complexity of \MinCSPwP$(R)$, where $R$ is restricted to a {\em symmetric} predicate, whose acceptance only depends on the number of true literals. This still includes many well-known CSPs including $r$LIN (which generalizes \textsc{MaxCut}), $r$SAT, $r$AND, $r$NAE, and $r$AE, where AE and NAE abbreviate All-Equal and Not-All-Equal respectively. 
Namely, the main contribution of this paper is the following theorem:

\begin{theorem}[Main Theorem (Informal)]\label{thm:dichotomy-informal}
    For any symmetric predicate $R$, \MinCSPwP($R$){} is either FPT or W[1]-hard.%
\end{theorem}

It turns out that among nontrivial predicates, $r$AND for any $r \geq 1$ and $2$AE (generalizing \textsc{MaxCut}) are the only ones in FPT. We design parameterized algorithms for them extending the previous algorithms for \textsc{MinCSP}. While many predicates inherit their hardness result from the hardness of \textsc{MinCSP}, we also prove hardness for the problems whose \textsc{MinCSP} is tractable, namely $r$AE for $r \geq 3$ and $\leq 1$-out-of-$r$SAT for $r \geq 2$. See Section~\ref{subsec:characterization} for our formal characterization. 

In the context of local search, a natural special case of \MinCSPwP($R$) is when the input $\mathcal{P}$ is restricted to be $\mathcal{C}_I(\beta)$ for some assignment $\beta : V \to \{ 0, 1 \}$. Theorem~\ref{thm:dichotomy-informal} addresses this special case as well, in the sense that our W[1]-hard reductions indeed have 
$\mathcal{P} = \mathcal{C}_I(\beta)$ for some assignment $\beta$.

\subsection{Related Work and Discussion}

The parameterized complexity of local search, whose goal is to improve the current solution to a strictly better solution nearby, was indeed studied in the parameterized complexity literature. Fellows, Fomin, Lokshtanov, Rosamond, Saurabh, and Villanger~\cite{fellows2012} study various problems like \textsc{$r$-Center}, \textsc{Vertex Cover}, \textsc{Odd Cycle Transversal}, \textsc{Max-Cut}, and \textsc{Max-Bisection} on special classes of sparse graphs like planar, minor-free, and bounded-degree graphs, 
and Szeider~\cite{szeider2011parameterized} considers a similar framework for {\textsc{SAT}} and {\textsc{MaxSAT}}.

This (strict) local search model can be relaxed to another model called \emph{permissive} local search~\cite{marx2011s}. Here we are given a solution $S$, and the goal is to find a solution $S'$ which has lower cost (which may differ from $S$ in more than $k$ vertices), or correctly conclude that every solution which differs from $S$ in at most $k$ vertices has higher cost. This model has been studied for a few problems - including a variant of stable marriage~\cite{marx2011s}, \textsc{Vertex Cover}~\cite{gaspers2012} and also in the context of {\textsc{Min Ones CSPs}}~\cite{krokhin2012hardness}, extending Marx's complete classification of the parameterized complexity of CSPs based on if there is a solution with exactly $k$ ones~\cite{marx2005parameterized}.

One crucial difference between our model and all previous models on CSPs is how the distance between the two solutions is defined. 
The previous papers define it as the number of variables with different values in two solutions, but our definition concerns the set of constraints whose satisfaction changes between two solutions. 
In fact, recall that our definition of a solution is an arbitrary set $\calP$ of constraints not necessarily feasible (i.e., corresponding to a variable assignment), which allows $\calP = \calC$ and captures algorithms parameterized by the number of unsatisfied constraints. In that sense, our model is slightly broader than traditional local search where one maintains a feasible solution and tries to strictly improve in every step.

Fixed parameter tractability of CSPs parameterized by the number of unsatisfied constraints has been an important and active research area~\cite{bonnet2016fixed, kim2021solving, kim2022directed, kim2023flow}, where a complete classification
for Boolean relations
was obtained very recently. Their definition of CSPs is strictly more general than ours, where each constraint language contains multiple predicates, even with different arities. 
It is an exciting research direction to extend our framework for every CSP.%

Our negative results can be interpreted as demonstrating the importance of {\em near-perfect instances} for classical FPT algorithms. For instance, we show  \MinCSPwP({2SAT}) is $W[1]$-hard while \MinCSP({2SAT}) is known to admit an FPT algorithm~\cite{razgon2008almost, raman2011paths, cygan2013multiway, narayanaswamy2012lp}, so the fact that the optimal solution indeed satisfies all but $k$ constraints affects the parameterized complexity. %

There are several natural research directions based on this paper. We believe that the following directions will be interesting to study further: 
\begin{enumerate}
    \item The most immediate one is to extend our characterization for all CSPs.
    One notable set of boolean relations whose complexity is still open in our model is $\{ x=y, x \wedge \neg y \}$.

    \item A natural variant of our model is when the solution is represented as an assignment to the variables, which is guaranteed to have a Hamming distance at most $k$ to a strictly better solution. For {\textsc{Min Ones CSP}}, where the goal is to find a satisfying assignment with the minimum number of true variables, this variant has been studied previously~\cite{krokhin2012hardness}.

    \item Other than CSPs, one framework that captures numerous results in the FPT literature is covering (resp. packing) problems, especially on graphs, where we want to select the smallest (resp. largest) subset $S \subseteq U$ subject to some covering (resp. packing) constraints. Local search naturally extends to this setting (i.e., we are given $S' \subseteq U$ with $|S' \Delta S| \leq k$), and it will be interesting to see if classical FPT results on graph covering/packing problems extend to this model, including \textsc{Multicut}~\cite{bousquet2011multicut, marx2011fixed}, \textsc{Min Bisection}~\cite{cygan2014minimum},
    \textsc{$k$-Cut}~\cite{kawarabayashi2011minimum},
    \textsc{Disjoint Paths}~\cite{robertson1995graph}. %

    \item In the approximation algorithms literature, there are numerous studies about CSPs on {\em structured instances}, most notably dense and expanding instances (see~\cite{anand2025complete} and references therein). It would be interesting to see whether these structures, if suitably parameterized, help the problems become more tractable in our model. 
\end{enumerate}

\section{Our Results}\label{sec:preliminaries}
\subsection{Constraint Satisfaction Problems}
In this section we recall basic definitions of constraint satisfaction problems.

For integer $r \in \mathbb{N}^+$, a \emph{boolean relation} $R$ is a subset of $\mathbb{F}_2^r$. The integer $r$ is called the \emph{arity} of $R$. For a boolean relation $R$ of arity $r$ and $b \in \mathbb{F}_2^r$, define the boolean relation $R \oplus b := \{\alpha \oplus b \mid \alpha \in R\},$ where $\oplus$ is the addition operator over $\mathbb{F}_2^r$. When considering $R$ as a clause in a CSP instance, $\oplus$ operator can be seen as variable negations. For a relation $R$ with index $[r]$ and $S \subseteq [r],$ define the projection relation $\pi_S(R) = \{\alpha_S \mid \alpha \in R\}$, where $\alpha_S$ means extracting the value on coordinates in $S$.%

A boolean relation $R$ of arity $r$ is \emph{symmetric} if for all $(x_1,\cdots,x_r) \in \mathbb{F}_2^r$ and any permutation $p: [r] \rightarrow [r]$, $(x_1,\cdots,x_r) \in R \iff (x_{p(1)}, \cdots, x_{p(r)}) \in R$. Equivalently, $R$ is symmetric if there exists $S \subseteq \{0,1,\cdots,r\}$ such that $(x_1,\cdots,x_r) \in R \iff \sum_{i=1}^r x_i \in S$ (where the addition is performed in $\mathbb{Z}$). We define $\SymRel(r, S)$ to be such $R$. %

A \emph{boolean constraint language} is a set of boolean relations. Two constraint languages are {\em equivalent} if they contain exactly the same set of relations. A \emph{clause} $C$ over a boolean constraint language $\Gamma$ is a pair $(V_C,R_C)$, where $R_C \in \Gamma$ and $V_C=(v_{C,1},\ldots,v_{C,r})$ is an $r$-tuple of variables where $r$ is the arity of $R_C$. We refer to $V_C$ as the \emph{scope} of the clause $C$. %
An assignment $\alpha: V_C \to \{0,1\}$ \emph{satisfies} clause $(V_C,R_C)$ if $(\alpha(v_{C,1}),\ldots,\alpha(v_{C,r})) \in R_C$, and $\alpha$ \emph{violates} the clause otherwise.

In this paper, we only consider a restricted set of boolean constraint languages, where the language $\Gamma$ includes all possible negation patterns $R \oplus b$ for some symmetric predicate $R$ of arity $r$ and $b \in \mathbb{F}_2^r$. We use the notation $\SymLang(R) = \{R \oplus b \mid b \in \mathbb{F}_2^r\}$ for such $\Gamma$, where $\SymLang$ stands for ``symmetric language with negation''. We further define $\SymLang(r,S)$ to be $\SymLang(\SymRel(r,S))$.

For a boolean constraint language $\Gamma$, a $\CSP(\Gamma)$ instance $I$ is a pair $(V_I, \mathcal{C}_I)$, where $V_I$ is the set of variables, and $\mathcal{C}_I$ is the set of clauses over the boolean constraint language $\Gamma$, all of which have scope inside $V_I$. We say a variable $v$ is incident to a clause $C$ if $v$ belongs to the scope of $C$.

For an assignment $\alpha: V_I \to \{0,1\}$, define $\mathcal{C}_I(\alpha)$ to be the set of clauses $\alpha$ satisfies in $I$. We further define $\cost_I(\alpha) = |\mathcal{C} \backslash \mathcal{C}_I(\alpha)|$, the number of unsatisfied clauses, to be the \emph{cost} of the assignment.
We define the \emph{size} of an instance $I$ as  $\max(|V_I|,|\mathcal{C}_I|)$ and denote it by $n_I$. 
We ignore the subscripts when the instance is clear in the context.

Now we give formal definitions for the problems we consider. We first define \MinCSP($\Gamma$), and then define our central problem \MinCSPwP($\Gamma$). In this paper, our main goal is to characterize the parameterized complexity of \MinCSPwP($\Gamma$) when $\Gamma = \SymLang(r,S)$ for some $r \in \mathbb{N}^+$ and $S \subseteq \{0,1,\cdots,r\}$.

\begin{framed}
\noindent\textbf{Problem:} \MinCSP($\Gamma$)

\noindent\textbf{Input:} A \CSP($\Gamma$) instance $I(V,\mathcal{C})$, an integer $k \in \mathbb{N}$.

\noindent\textbf{Parameter:} $k$.

\noindent\textbf{Output:} An assignment $\alpha$ with $\cost(\alpha) \le k$ if it exists, otherwise report that such an assignment does not exist.
\end{framed}

\begin{framed}
\noindent\textbf{Problem:} \MinCSPwP($\Gamma$)

\noindent\textbf{Input:} A \CSP($\Gamma$) instance $I(V,\mathcal{C})$, an integer $k \in \mathbb{N}$, and $\mathcal{P} \subseteq \mathcal{C}$

\noindent\textbf{Parameter:} $k$

\noindent\textbf{Promise: } There exists an assignment $\alpha_0: V_I \rightarrow \{0,1\}$  such that $|\mathcal{C}_I(\alpha_0) \Delta \mathcal{P}| \le k$.%

\noindent\textbf{Output:} A {\em good} assignment $\alpha$,
where the word {\em good} simply means $|\mathcal{C}_I(\alpha)| \geq 
|\mathcal{C}_I(\alpha_0)|$.
(We do not require the output $\alpha$ to satisfy $|\mathcal{C}_I(\alpha) \Delta \mathcal{P}| \le k$.)
\end{framed}

From the definition, \MinCSP($\Gamma$) is a special case of \MinCSPwP($\Gamma$) by setting $\mathcal{P}$ to be the constraint set $\mathcal{C}$.

A boolean constraint language $\Gamma$ is \emph{trivial} if $\Gamma = \varnothing$ or $\Gamma = \{\{0,1\}^r\}$ for some $r \in \mathbb{N}^+$ and \emph{nontrivial} otherwise. When $\Gamma$ is trivial, \MinCSP($\Gamma$) and \MinCSPwP($\Gamma$) can be solved by trivial algorithms. So in the following we only consider nontrivial boolean constraint languages.

\subsection{Abbreviations for Special CSPs}

For simplicity, we abbreviate boolean constraint languages mentioned in the paper. Here we introduce the following simple but useful claim in identifying equivalence between boolean constraint languages.
\begin{claims}
    $\SymLang(r,S)$ and $\SymLang(r,\{r-x \mid x \in S\})$ are equivalent boolean constraint languages.
\end{claims}
\begin{proof}
    From $\SymRel(r,S) + (1,1,\cdots,1) = \SymRel(r,\{r-x \mid x \in S\})$, where $(1,1,\cdots,1)$ is the element in $\mathbb{F}_2^r$ that has value 1 in each coordinate,  we have $\SymRel(r,S)+b = \SymRel(r,\{r-x \mid x \in S\}) + (b + (1,1,\cdots,1))$ for any $b \in \mathbb{F}_2^r.$
\end{proof}

We will use the following abbreviations in this paper.

\begin{itemize}
    \item $r\AND = \SymLang(r,\{0\}) = \SymLang(r,\{r\})$.
    \item $r\SAT = \SymLang(r,\{1,2,\cdots,r\}) = \SymLang(r,\{0,1,\cdots,r-1\})$.
    \item $r\ALLEQ$ (All Equal) = $\SymLang(r, \{0,r\})$.
    \item $\LE$ = $\SymLang(r, \{0,1\}) = \SymLang(r, \{r-1,r\})$.
\end{itemize}

We also use abbreviations for the corresponding \textsc{MinCSP} and \textsc{ImproveMaxCSP} problems. For example, \MinCSPwP($r$\AND) corresponds to the problem \MinCSPwP$(\SymLang(r,\{0\}))$.

\subsection{Our Characterization}
\label{subsec:characterization}
Having defined our framework, we first provide the formal version of \Cref{thm:dichotomy-informal}:

\begin{theorem}[Main Theorem (Formal)] \label{thm:dichotomy-formal}
    For any $r \in \mathbb{N}^+$ and $S \subseteq \{0,1,\cdots,r\}$, the problem \MinCSPwP$(\SymLang(r,S))$ is either FPT or W[1]-hard.
\end{theorem}

As $\MinCSP(\Gamma)$ is a special case of \MinCSPwP$(\Gamma)$, the W[1]-hardness of the former implies that of the latter. The recent characterization of parameterized tractability of $\MinCSP(\Gamma)$~\cite{kim2023flow} implies the following theorem, which is proved in~\Cref{sec:hardness-from-fpt}. 

\begin{theorem}\label{characterization:main-hardfrommincsp}
    For nontrivial $\Gamma = \SymLang(r,S)$ for some $r \in \mathbb{N}^+$ and $S \subseteq \{0,1,\cdots,r\}$ that is not equivalent to $r\AND, r\ALLEQ$ or $\LE$, \MinCSPwP$(\Gamma)$ is W[1]-hard.
\end{theorem}

It remains to study the remaining problems. 
\begin{remark}\label{remarkalgo-boundary}
We say that an algorithm $A$ {\em solves} \MinCSPwP$(\Gamma)$, if for any instance that satisfies the promise, $A$ outputs a good assignment. However, the run time is measured on all possible instances. In other words, $A$ runs in $f(k)n^{O(1)}$ time if it terminates and outputs some assignment in $f(k)n^{O(1)}$ time no matter whether the input instance $(I,k,\mathcal{P})$ satisfies the promise. Clearly, only the run time on instances satisfying the promise is critical, as we can always set a run time limit to the algorithm to rule out instances that require too much time to indicate their violation of promise.
\end{remark}
For \MinCSPwP$(r \AND)$, a color-coding combined with the LP rounding for {\sc Densest Subhypergraph} yields the following result proved in Section~\ref{sec:and}. 

\begin{theorem}\label{characterization:and}
    There exists an algorithm that solves \MinCSPwP($r$\AND) in $2^{O(k \log k)} n^{O(1)}$ time for all $r \ge 1$.
\end{theorem}

For $r\ALLEQ$, its tractability depends on $r$. When $r = 2$, we use the celebrated {\em unbreakability} framework~\cite{kawarabayashi2011minimum, chitnis2016designing, cygan2014minimum, lokshtanov2022parameterized} to design a parameterized algorithm in Section~\ref{sec:2ae}.

\begin{theorem}\label{characterization:2ae}
There exists an algorithm that solves \MinCSPwP($2$\ALLEQ) in time $2^{O(k^2)}n^{O(1)}$. 
\end{theorem}

However, for $r\ALLEQ$ with $r \geq 3$, we show nontrivial reductions from Paired Minimum $st$-Cut to show the following hardness in~\Cref{sec:ae}.

\begin{theorem}\label{characterization:3ae}
    \MinCSPwP$(r\ALLEQ)$ is W[1]-hard for all $r \ge 3$.
\end{theorem}

Finally, for $\LE$, we prove that even the first nontrivial case $r = 2$, which is equivalent to $2\SAT$, is already W[1]-hard (\Cref{sec:le}). Even though $\MinCSP(2\SAT)$ is in FPT, this result follows from the (slightly informal) fact that {\em Vertex Cover and Independent Set are equivalent in our model}. 

\begin{theorem}\label{characterization:le}
    \MinCSPwP$(\LE)$ is W[1]-hard for any $r \ge 2$.
\end{theorem}

It is easy to see that the other theorems imply \Cref{thm:dichotomy-formal} in a straightforward way.

\section{FPT Algorithm for $r$AND}\label{sec:and}

The main goal of this section is to prove \Cref{characterization:and}. During the course of the
algorithm, we create instances involving \AND\ constraints of different arities. Hence, we will show a slightly stronger result:
\begin{theorem}\label{thm:main-and}
    There is an algorithm that solves \MinCSPwP($\cup_{r' \le r} r'\AND$) in $2^{O(k \log k)} n^{O(1)}$ time for all $r \ge 1$.
\end{theorem}
Since $\cup_{r' \le r} r'\AND \supseteq r\AND$, \Cref{thm:main-and} immediately shows that for any $r \ge 1$, the problem $\MinCSPwP(r\AND)$ is FPT.

\subsection{Preliminaries}

In this section we recall basic definitions about hypergraphs. A \emph{hypergraph} $H(V,E)$ is a generalization of graph where an edge can contain multiple vertices. For $e \in E$ and $v \in V$, say $v$ is incident to $e$ and $e$ is incident to $v$ if $e$ contains $v$. For a hypergraph $H$ and $V_0 \subseteq V$, we denote $H(V_0)$ as the \emph{induced subhypergraph} 
of $V_0$ in $H$, and denote $E(H(V_0))$ to be the edge set of the induced subhypergraph.

In the main algorithm in this section (provided in \Cref{thm:main-andalgo}), we need to solve the following problem: 
\begin{framed}
  \noindent\textbf{Problem:} \MISVW.

    \noindent
    \textbf{Input:} A hypergraph $H(V,E)$ and a weight function $w: V \to \mathbb{Z}$.
    
    \noindent
    \textbf{Output:} A vertex subset $V_0 \subseteq V$ that maximizes $|E(H(V_0))| - \sum_{v \in V_0} w(v)$.
\end{framed}
The problem is closely related to {\sc Densest Subgraph} (See \cite{lanciano2024survey} for a survey) and has a polynomial-time algorithm based on similar algorithms for {\sc Densest Subgraph}.
One algorithm is to observe that $|E(H(V_0))| - \sum_{v \in V_0} w(v)$ is supermodular and run a submodular minimization algorithm.\footnote{We thank an anonymous reviewer for this observation.} The below is an LP-based algorithm specific to our setting. 

\begin{lemma}\label{lemma:densestsubgraph}
    There exists a polynomial-time exact algorithm that solves \MISVW.
\end{lemma}
\begin{proof}
    Consider the following LP: \begin{equation*}
\begin{array}{ll@{}ll}
\text{maximize}  & \sum_{e \in E} x_e - \sum_{v \in V} w(v)y_v &\\
\text{subject to}& x_e \le y_v & \forall e \in E, v \in V\ \mathrm{s.t.\ }e\mathrm{\ is\ incident\ to\ }v\\
                 & x_e \in [0,1]   & \forall e \in E \\
                 & y_v \in [0,1]   & \forall v \in V
\end{array}
\end{equation*}
Since we use $x_e \le y_v$ to mimic the constraint that ``one should select a vertex before selecting incident edges'', if the LP is integral (i.e. $x_e, y_v \in \{0,1\}$) then the set of vertices $v$ with $y_v=1$ in the optimal solution of the LP corresponds to the correct output to the \MISVW problem.

The algorithm solves the LP, providing optimal fractional solution $(x^\star_e, y^\star_v)$, and rounds the solution to an integer solution in the following way: Arbitrarily select $p \in (0,1]$ and set $x_e^p = [x^\star_e \ge p], y_v^p = [y^\star_v \ge p]$, where $[P]$ denotes the Iverson bracket and equals to $1$ if $P$ is true otherwise $0$ for some statement $P$.
The solution is clearly an integral solution to the LP.

Now we show that any such $p$ produces an integral solution that has the same value as that of $(x^\star_e,y^\star_v)$, denoted as $f^\star$. Define $f(p)$ to be the objective function value of the solution $(x_e^p, y_v^p)$. Clearly we have $f(p) \le f^\star$. We further have \begin{equation}
    \begin{aligned}
        \int_0^1 f(p) \mathrm{d} p & = \int_0^1 \left(\sum_{e \in E} [x_e^\star \ge p] - \sum_{v \in V} w(v) [y_v^\star \ge p] \right) \mathrm{d} p \\
        & = \sum_{e \in E} x_e^\star - \sum_{v \in V} w(v) y_v^\star = f^\star
    \end{aligned}
\end{equation}
So the set $\{x \in (0,1] \mid f(x) < f^\star\}$ has zero measure. Further from the fact that $f(p)$ is a step function of finite number of ``steps'' and each ``step'' has nonzero length, such set must be empty, and $f(p) = f^\star$.
\end{proof}

Our FPT algorithms in \Cref{sec:and} and \Cref{sec:2ae} will use the technique of {\em color coding} and its derandomization given below.

\begin{theorem} [\cite{naor1995splitters}] Given a set $U$ of size $n$ and integers $0 \leq a, b \leq n$, one can in time $2^{O(\min(a, b)\log(a+b))} \cdot  n \log n$ construct a family $\calF$ of at most $2^{O(\min(a, b) \log(a + b))} \log n$ 
colorings $U \to \{0, 1\}$ such that the following holds: for any sets $A, B \subseteq U$, $A \cap B = \emptyset$, $|A| \leq a$, $|B| \leq b$, there exists a coloring $\chi \in\mathcal{F}$ with $\chi(a) = 1$ for all $a \in A$ and $\chi(b) = 0$ for all $b \in B$.
\label{thm:splitter}
\end{theorem}

\subsection{Algorithm with Variable Solution}

Given an instance $(I,k,\mathcal{P})$, we say that $\mathcal{P}$ is {\em satisfiable} if some assignment satisfies all clauses in $\mathcal{P}$. We will give an algorithm that runs in $2^{O(k \log k)} n^{O(1)}$ time and solves $\MinCSPwP(\cup_{r' \le r} r'\AND)$ when the instance satisfies an additional promise that $\mathcal{P}$ is satisfiable.

When this additional promise is satisfied, it is easy to find an assignment satisfying all the clauses of $\mathcal{P}$ in polynomial time, so in the following we assume that the instance $(I,k,\mathcal{P})$ is equipped with such an assignment $\alpha$. In the next section, we show how to transform any instance into instances where this additional promise is true.

Equipped with this additional assignment, we then restructure the instance in the following way: Define $k' := k + |\mathcal{P} \Delta C_I(\alpha)|$ and $\mathcal{P}':= C_I(\alpha)$, and replace $(I,k,\mathcal{P},\alpha)$ with $(I,k',\mathcal{P}',\alpha)$, to keep $\mathcal{P}$ to be maximal in accordance with $\alpha$. In other words, we reset $\mathcal{P}$ to be the set of clauses satisfied by $\alpha$.
If the instance satisfies the promise (that some good assignment is close to the $\mathcal{P}$), the good assignment satisfies at most $|\mathcal{P}| + k$ clauses, which implies that $|\mathcal{P} \Delta C_I(\alpha)| \le k$, and $k' \le 2k$. If $k' > 2k$, we return an arbitrary assignment (to handle cases where the promise is not met). %

Now we show that, there exists an good assignment close to $\mathcal{P}$ whose Hamming distance to $\alpha$ is small.

\begin{lemma}\label{lem:smallhamdist}
    For a $\MinCSPwP(\cup_{r' \le r} r'\AND)$ instance $(I,k,\mathcal{P},\alpha)$, there exists a good assignment $\beta: V \to \{0,1\}$ with $|C_I(\beta) \Delta \mathcal{P}| \le k$ and $|\{v \in V \mid \alpha(v) \ne \beta(v) \}| \le rk$.
\end{lemma}
\begin{proof}
    Choose a good $\beta$ with $|C_I(\beta) \Delta \mathcal{P}| \le k$, and if there are multiple ones, choose the one with the smallest hamming distance to $\alpha$.
    
    For any $v$ with $\alpha(v) \ne \beta(v)$, consider assignment $\beta'$ that differs from $\beta$ only on $v$. $\beta'$ has smaller Hamming distance to $\alpha$ than $\beta$, so either $\beta'$ is not good, or $\beta'$ is good but does not satisfy the promise. Both cases require that some clause $C \in \mathcal{C}$ is satisfied in $\beta$ but violated in $\beta'$. Since all clauses are AND clauses and $\beta'(v) = \alpha(v)$, such clause $C$ is violated in $\alpha$.
    Therefore, all variables contributing to the hamming distance are incident to some clauses in $C_I(\beta) \Delta \mathcal{P}$, and the statement follows from the fact that $|C_I(\beta) \Delta \mathcal{P}| \le k$ and that the arity of all relations in $\cup_{r' \le r} r'\AND$ is no more than $r$.
\end{proof}

Now we use the (derandomized) color coding technique along with the algorithm of \Cref{lemma:densestsubgraph} %
for~\MISVW to solve our problem 
\MinCSPwP$(\cup_{r' \le r} r'\AND)$.

\begin{theorem}\label{thm:main-andalgo}
   There is an algorithm solving \MinCSPwP($\cup_{r' \le r} r'\AND$) in $2^{O(k \log k)} n^{O(1)}$ time when the input instance satisfies an additional promise that $\mathcal{P}$ is satisfiable.
\end{theorem}
\begin{proof}
    First use \Cref{thm:splitter} with $a = b = rk$ to obtain a family of at most $2^{O(k \log k)} \log n$ colorings with binary labels for variables. The algorithm then produces an assignment for each coloring and reports the best among them. In the following we fix a particular coloring.
    
    Let $L_1$ be the set of variables with label $1$. Construct a binary relation $\sim$ on $L_1$, where $u \sim v$ if $u=v$ or some clause $C \in \mathcal{P}$ is incident to both $u$ and $v$. Since $\sim$ is reflexive and symmetric, its transitive closure is an equivalence relation. Denote $L_1 /\sim$ as the quotient set of the equivalence relation which includes all equivalence classes of the transitive closure. The intuition is that, the algorithm will produce the answer from $\alpha$ by flipping several variables with label 1, and variables in one equivalence class are considered a group and would be flipped or kept at the same time.

    We then construct a hypergraph $H(V_H,E_H)$ with weight function $w: V_H \to \mathbb{Z}$ as follows: \begin{itemize}
        \item For each equivalence class in $L_1 / \sim$, introduce a vertex $v$ in $V_H$, whose weight $w(v)$ equals to the number of clauses in $\mathcal{P}$ incident to any variable in the equivalence class. Intuitively, the weight measures that how many clauses we will lose if we flip this equivalence class. %
        \item For each clause $C \in \mathcal{C}_I \backslash \mathcal{P}$, if $C$ can be satisfied by flipping $L_1$ from $\alpha$, introduce a hyperedge incident to all equivalence classes containing incident vertices of $C$. So the clause will be satisfied if we flip all equivalence classes incident to this hyperedge.
    \end{itemize}

    Finally, the algorithm uses the algorithm for \MISVW proposed in \Cref{lemma:densestsubgraph} to find $V_0 \subseteq V_H$ that maximizes $|E(H(V_0))| - \sum_{v \in V_0} w(v)$, and returns the assignment produced by flipping in $\alpha$ all variables belonging to equivalence classes in $V_0$.

    The algorithm runs in $2^{O(k \log k)} n^{O(1)}$ time since the whole process can be done in polynomial time except trying all $2^{O(k \log k)} n^{O(1)}$ colorings produced by \Cref{thm:splitter}. Now we show that the algorithm produces a good assignment for instances satisfying the promise that some good assignment is close to the $\mathcal{P}$. Since we finally choose the best assignment among all colorings, it is sufficient to show that there exists some coloring that produces a good assignment.

    Fix any good assignment $\beta$ with $|C_I(\beta) \Delta \mathcal{P}| \le k$ and $|\{v \in V \mid \alpha(v) \ne \beta(v) \}| \le rk$, where $\alpha$ is an assighnment such that $\calP = \calC_I(\alpha_0)$.
    According to \Cref{lem:smallhamdist} such $\beta$ exists. Define $N(C_I(\beta) \Delta \mathcal{P})$ to be all variables incident to any clause in $C_I(\beta) \Delta \mathcal{P}$. Define $V_0 = \{v \in N(C_I(\beta) \Delta \mathcal{P}) \mid  \alpha(v) = \beta(v)\}$ and $V_1 = N(C_I(\beta) \Delta \mathcal{P}) \backslash V_0$. Note that both $|V_0|,|V_1| \leq rk$.%

    According to \Cref{thm:splitter}, there exists a coloring assigning $V_0$ to $0$ and $V_1$ to $1$. We consider the behavior of the algorithm when trying this coloring.
    
    We first claim that, there exists a set of equivalence classes in $L_1 / \sim$ whose union is exactly $V_1$. This is equivalent to the statement that there is no $u \in V_1$ and $v \in L_1 \backslash V_1$ such that $u \sim v$. Suppose $u \sim v$ and $u \in V_1$, then some $C \in \mathcal{P}$ is incident to both of them. Since $C$ is incident to $u \in V_1$ which satisfies $\alpha(u) \ne \beta(u)$, $C$ is violated by $\beta$, so $C \in \mathcal{C}_I(\beta) \Delta \mathcal{P}$. Then $v \in N(C_I(\beta) \Delta \mathcal{P})$, which finishes the proof since it implies $v \not\in L_1 \backslash V_1$.

    The previous statement shows that, there exists some $V^\star \subseteq V_H$, such that by flipping from $\alpha$ all variables in equivalence classes corresponding to $V^\star$, one can produce $\beta$. We further show the following two claims: \begin{itemize}
        \item For any $V \subseteq V_H$, suppose $\alpha'$ is produced by flipping from $\alpha$ all variables in every equivalence class in $V$, then $\cost_I(\alpha) - \cost_I(\alpha') \le |E(H(V))| - \sum_{v \in V} w(v)$. Notice that the latter is the objective of the \MISVW problem. In other words, the cost of $\alpha'$ would not be underestimated. To see this, consider clauses in $C_I(\alpha) \Delta C_I(\alpha')$. Clauses satisfied in $\alpha$ but unsatisfied in $\alpha'$ are those belonging to $\mathcal{P}$ and incident to $V$, counted exactly by $\sum_{v \in V} w(v)$. Clauses satisfied in $\alpha'$ but unsatisfied in $\alpha$ are partially counted by $|E(H(V))|$, since it does not include satisfied clauses in $\alpha'$ violated by the assignment produced from flipping $L_1$ in $\alpha$.
        However, all hyperedges in $E(H(V))$ correspond to clauses satisfied in $\alpha'$ but unsatisfied in $\alpha$.
        \item $\cost_I(\alpha) - \cost_I(\beta) = |E(H(V^\star))| - \sum_{v \in V^\star} w(v)$, i.e. the cost of $\beta$ is correctly estimated by the \MISVW objective. We only need to show that $E(H(V^\star))$ contains all clauses satisfied in $\beta$ but unsatisfied in $\alpha$. Since the coloring colors $V_0$ with $0$ and $V_1$ with $1$, all clauses satisfied in $\beta$ but unsatisfied in $\alpha$ have label-1 incident variables inside $V_1 \subseteq L_1$, so the clauses will introduce hyperedges in $H$ whose scope is inside $V^\star$, so they are all counted by $|E(H(V^\star))|$.
    \end{itemize}
    
    The previous two claims imply that, $V^\star$ is one of the optimal solutions to the \MISVW problem for the constructed hypergraph $H$, and any optimal solution of the \MISVW problem corresponds to a good assignment to the CSP instance, which finishes the proof.
\end{proof}

\subsection{From Clause Solution to Variable Solution for $r$AND}
\label{sec:clausetovarand}

Now we show how to solve \MinCSPwP$(\cup_{r' \le r} r'\AND)$ using \Cref{thm:main-andalgo}. We first define an auxiliary operation, which essentially fixes a given variable to a fixed boolean value $a \in \{0,1\}$.%

\begin{definition}[Value Assignment]\label{def:assignvalue}
    For a \MinCSPwP$(\cup_{r' \le r} r'\AND)$ instance $(I(V,\mathcal{C}),k,\mathcal{P})$, a variable $v \in V$ and a boolean value $a \in \{0,1\}$, to ``assign $v$ to $a$ in $(I,k,\mathcal{P})$'', we mean producing a \MinCSPwP$(\cup_{r' \le r} r'\AND)$ instance $(I'(V',\mathcal{C}'),k',\mathcal{P}')$ in the following way:%
\begin{itemize}
    \item $V' = V \backslash v$. Initially, $k' = k$.
    \item For all clauses $C(V_C,R_C) \in \mathcal{C}$, define $C' = (V'_C = V_C \backslash \{v\}, R'_C = \{r_{V_C \backslash \{v\}} \mid r \in R, r(v) = a\})$ if $v \in V_C$ and otherwise $C' = C$. If $R' = \varnothing$ and $C \in \mathcal{P}$ or $R' = \{0,1\}^{r-1}$ and $C \not\in \mathcal{P}$, decrease $k'$ by $1$. If $R'$ is not trivial, add $C'$ to $\mathcal{C}'$, and add $C'$ to $\mathcal{P}'$ if $C \in \mathcal{P}$.
\end{itemize}
Notice that the instance $(I',k',\mathcal{P}')$ may not satisfy the promise that there exists a good assignment $\alpha_0$ with $|\mathcal{C}_{I'}(\alpha_0) \Delta \mathcal{P}'| \le k'$. %
\end{definition}

The following claims are straightforward from the definition.
\begin{claims}\label{claim:keeppromise}
    Suppose instance $(I', k', \mathcal{P}')$ is produced from assigning $v$ to $a$ in $(I,k,\mathcal{P})$, and there exists an assignment $\alpha$ with the least cost for $I$ which satisfies $\alpha(v) = a$ and $|C_I(\alpha) \Delta \mathcal{P}| \le k$, then $\alpha_{V'}$ has the least cost for $I'$ and $|C_{I'}(\alpha_{V'}) \Delta \mathcal{P}'| \le k'.$
\end{claims}
\begin{claims}\label{claim:keepvalue}
    Suppose instance $(I', k', \mathcal{P}')$ is produced from assigning $v$ to $a$ in $(I,k,\mathcal{P})$, then for all assignments $\alpha, \beta: V_I \to \{0,1\}$ with $\alpha(v) = \beta(v) = a$, $\cost_I(\alpha) - \cost_{I'}(\alpha_{V'}) = \cost_I(\beta) - \cost_{I'}(\beta_{V'})$.
\end{claims}

Now we show the algorithm, which branches on variables involved in conflicts in the solution to produce instances in which some assignment satisfies the solution. Notice that since we use \Cref{def:assignvalue} to do branching, some  \MinCSPwP$(\cup_{r' \le r} r'\AND)$ instances which do not satisfy the promise that some good solution is close to the current solution may be created.

\begin{theorem}\label{thm:variableprediction}
    Suppose there exists an algorithm $A$ that runs in $2^{O(k \log k)} n^{O(1)}$ time, solves \MinCSPwP$(\cup_{r' \le r} r'\AND)$ when the input instance to $A$ satisfies an additional promise that $\mathcal{P}$ is satisfiable. %
    Then there is an algorithm that solves \MinCSPwP$(\cup_{r' \le r} r'\AND)$ in $2^{O(k \log k)} n^{O(1)}$ time (with calls to algorithm $A$). %
\end{theorem}
\begin{proof}
     Given a \MinCSPwP$(\cup_{r' \le r} r'\AND)$ instance $(I, k, \mathcal{P})$, consider the following algorithm:
     \begin{itemize}
        \item If $k \le 0$, return any assignment (to handle instances violating the promise).
        \item Otherwise, if there exists $v \in V$ and $C_1,C_2 \in \mathcal{P}$ such that $C_1$ contains literal $v$ while $C_2$ contains literal $\lnot v$, branch on variable $v$: Produce two instances $(I_0,k_0,\mathcal{P}_0)$ and $(I_1,k_1,\mathcal{P}_1)$ by assigning $v$ to $0$ respectively $1$, and solve two instances, producing assignments $\alpha_0$ and $\alpha_1$. Extend two assignments to $V$ with $\alpha_0(v) = 0$ and $\alpha_1(v) = 1$ and return the assignment with smaller cost in $I$.
        \item Otherwise, run algorithm $A$ to solve the instance $(I,k,\mathcal{P})$.
    \end{itemize}
    
    We do induction on $k$ to prove the correctness of the algorithm. In the base case, when $k < 0$, the instance cannot satisfy the promise. Now suppose the statement holds for all integer less than $k$.
    \begin{itemize}
        \item If there exists $v \in V$ and $C_1,C_2 \in \mathcal{P}$ such that $C_1$ contains literal $v$ while $C_2$ contains literal $\lnot v$, then any assignment for $I$ would falsify one of $C_1$ and $C_2$. So the two produced instances $(I_0,k_0,\mathcal{P}_0)$ and $(I_1,k_1,\mathcal{P}_1)$ have $k_0 < k, k_1 < k$, and the algorithm would finally output some assignments for them from the induction hypothesis. Further, if $(I,k,\mathcal{P})$ satisfies the promise, there exists a good $\alpha$ such that $|C_I(\alpha) \Delta \mathcal{P}| \le k$. Without loss of generality suppose $\alpha(v) = 0$, then from \Cref{claim:keeppromise} $(I_0, k_0, \mathcal{P}_0)$ satisfies the promise, and the algorithm would produce the good assignment for $(I_0,k_0,\mathcal{P}_0)$. Further according to \Cref{claim:keepvalue}, extending the assignment with $\alpha(v) = 0$ results in a good assignment for $(I,k,\mathcal{P}).$
        \item Otherwise, for all $v \in V$, at most one of the literals $v$ and $\lnot v$ appears in $\mathcal{P}$. Assigning variables to satisfy the appeared literal, we end up with an assignment satisfying all clauses in $\mathcal{P}$. So the correctness of the algorithm follows from $A$.
    \end{itemize}

    The runtime of the algorithm follows the recursion \begin{equation}
        T(k) = \max\left(2\max_{0 \le i \le k-1}T(i) + n^{O(1)}, 2^{O(k\log k)} n^{O(1)}\right)
    \end{equation}
    since the instance size never increases within the recursion. $T(k) \le 2^{O(k \log k)} n^{O(1)}$ satisfies the recursion.
\end{proof}

\begin{proof} [Proof of~\Cref{thm:main-and}]
    Combining the algorithms in \Cref{thm:main-andalgo} and \Cref{thm:variableprediction}, we get the algorithm that solves \MinCSPwP$(\cup_{r' \le r} r'\AND)$ in $2^{O(k \log k)} n^{O(1)}$ time.
\end{proof}

\section{FPT Algorithms for $2\ALLEQ$}\label{sec:2ae}

In this section we prove \Cref{characterization:2ae}. 
\subsection{Preliminaries}

We recall useful definitions about graphs. Let $G=(V,E)$ be a multigraph, where parallel edges and loops are allowed. For two disjoint vertex sets $A,B \subseteq V$, we denote by $E(A,B)$ the set of edges between $A$ and $B$. When $(A,B)$ is a partition, then $E(A,B)$ is the set of edges in the corresponding cut. For a vertex subset $V' \subseteq V$, we denote by $G(V')$ the subgraph of $G$ induced by $V'\,.$ We write $E(G(V'))$ for the set of edges in $G(V')\,.$  For a set $E'\subseteq E\,,$ we denote by $G\setminus E'$ the multigraph $(V, E\setminus E')\,.$ For $S \subseteq V$, we let $N(S) := \{ v \in V \setminus S : \exists u \in S \mbox{ such that } (u, v) \in E \}$ be the open neighborhood of $S$. We say two cuts $(A, B)$ and $(X, Y)$ are $k$-close if $|E(A, B) \Delta E(X, Y)| \leq k$. When the graph is clear in the context, we define $n = \max(|V|,|E|)$ to be the size of the graph.

Notice that we can see clauses as edges connecting variables, assigning boolean values to variables as partitioning the variable set into two sets, and the equality and non-equality constraints correspond to uncut and cut constraints respectively. This motivates us to introduce the following graph problem.
Given a graph $G = (V, E)$, an edge type function $t : E \to \{0, 1 \}$ and a partition $(A,B)$ of $V$, define $C(A,B) = (t^{-1}(1) \cap E(A,B)) \cup (t^{-1}(0) \cap (E \backslash E(A,B))$ to be the set of edges satisfied by $(A,B)$, where edges of type $1$ need to be separated, while edges of type $0$ should be inside some vertex set.

\begin{framed}
    \noindent\textbf{Problem:} \CUT
    
    \noindent\textbf{Input:} A connected multigraph $G(V,E)$, an integer $k$, $\mathcal{P} \subseteq E$, and an edge type function $t: E \to \{0,1\}$.

    \noindent\textbf{Parameter:} $k$

    \noindent\textbf{Promise:}  There exists some $(A,B)$ that maximizes $|C(A,B)|$ and satisfies $|C(A,B) \Delta \mathcal{P}| \le k.$

    \noindent\textbf{Output:} Any partition $(A,B)$ of vertices that maximizes $|C(A,B)|$. (Notice that we do not require $|C(A,B) \Delta \mathcal{P}| \le k$ for the output.)
\end{framed}
Clearly, if the given graph for \CUT is not connected, we can solve each connected component separately, so connectivity of $G$ does not lose any generality for the problem. The following claim is straightforward.
\begin{claims}\label{claim:2ae=cut}
    $\MinCSPwP(2\ALLEQ)$ and \CUT are equivalent.
\end{claims}

In the following, we present the algorithm for \CUT.

\subsection{From Edge Solution to Vertex Solution}

We first apply a pre-processing algorithm to the given \CUT\ instance, intending to make it in accordance with some partition of the vertex set.

\begin{lemma}\label{lem:2ae-edge2partition}
    There exists an algorithm that given a \CUT\ instance, runs in $4^k n^{O(1)}$ time and finds a bipartition $(A,B)$ of the vertex set $V$ such that there exists a \emph{good} partition $(A^\star,B^\star)$ which satisfies $|C(A^\star,B^\star) \Delta \mathcal{P}| \le k$ and $|C(A^\star,B^\star) \Delta C(A,B)| \le 3k$.
\end{lemma}
\begin{proof}
    Fix any $(A^\star,B^\star)$ satisfying $|C(A^\star,B^\star) \Delta \mathcal{P}| \le k$. We first find the largest subset $\mathcal{P}_1 \sse \mathcal{P}$ that can be simutaneously satisfied by some bipartition using a $4^kn^{O(1)}$-time algorithm for $\MinCSP(2\ALLEQ)$~\cite{dabrowski2023almost}: Since $C(A^\star,B^\star) \cap \mathcal{P}$ is one possible solution, this means that $|\mathcal{P}_1| \geq |C(A^\star,B^\star) \cap \mathcal{P}| \geq \max(\mathcal{P}, |C(A^\star,B^\star)|) - k$. Let $(A,B)$ be any bipartition satisfying $\mathcal{P}_1$. Consider $C(A,B)$: It contains all edges in $\mathcal{P}_1$ and perhaps more. However, since the size cannot be larger than $|C(A^\star,B^\star)|$, there are at most $k$ edges in $C(A,B) \setminus \mathcal{P}_1$. Now we have \[ |C(A,B) \Delta C(A^\star,B^\star)| \leq |\mathcal{P} \Delta C(A^\star,B^\star)| + |\mathcal{P} \setminus \mathcal{P}_1| + |C(A,B) \setminus \mathcal{P}_1| \leq 3k. \] This completes the proof. \end{proof}
    
From \Cref{lem:2ae-edge2partition}, we can set $\mathcal{P}$ to be such $C(A,B)$ and triple $k$ to change the original \CUT\ instance to an equivalent \CUT\ instance, whose $\mathcal{P}$ is given by some $C(A,B)$ for a vertex partition $(A,B)$.

\subsection{Algorithms for Instances with Vertex Solution}
In the rest of this section, we prove the following theorem:

\begin{theorem}\label{thm:main-cut}
There exists an algorithm that, given a \CUT\ instance $(G(V,E),k,\mathcal{P},t)$ where $\mathcal{P} = C(A,B)$ for some partition $(A,B)$ of $V$, finds a partition $(X,Y)$ that maximizes $|C(X,Y)|$ in time $2^{O(k^2)} n^{O(1)}$. 
\end{theorem}

We first show how this theorem implies the main theorem. 

\begin{proof}[Proof of~\Cref{characterization:2ae}]
    According to \Cref{claim:2ae=cut}, we only need to find such algorithm for \CUT. For a \CUT\ instance, we first apply the algorithm stated in \Cref{lem:2ae-edge2partition} to produce an instance equivalent to the input whose $\mathcal{P}$ is given by some $C(A,B)$ for a vertex partition. Then we apply the algorithm stated in \Cref{thm:main-cut} to solve the new instance. The runtime is $4^kn^{O(1)} + 2^{O(k^2)} n^{O(1)} = 2^{O(k^2)} n^{O(1)}$.
\end{proof}

Recall that we are given a vertex partition $(A,B)$ as the solution, while there exists a good partition $(X,Y)$ with $|C(A, B) \Delta C(X, Y)| \leq k$. Let $AX := A \cap X$ and define $AY, BX, BY$ similarly. 

Our high-level approach follows the framework first introduced by Kawarabayashi and Thorup~\cite{kawarabayashi2011minimum} for Minimum $k$-cut and its refinement by Chitnis et al.~\cite{chitnis2016designing}. 
Let us define the following {\em terminal version} of the problem. 
We will fix $k$ to be the parameter of the initial instance, which \emph{remains invariant throughout the recursive calls to the terminal version}.

\begin{mdframed}[nobreak=true]
    \noindent\textbf{Problem:} \CUTTERMINAL

    \noindent\textbf{Input:} A connected multigraph $G(V,E)$, an integer $k' \le k$, a $\mathcal{P} \subseteq E$, an edge type function $t: E \to \{0,1\}$, a set of terminals $T \subseteq V$ with $|T| \le 2k$ (Notice that it is $k$ rather than $k'$) and a \emph{marked edge set} $M \subseteq E$.

    \noindent\textbf{Parameter:} $k'$

    \noindent\textbf{Promise:} \begin{itemize}
        \item $\mathcal{P} = C(A,B)$ for some partition $(A,B)$;
        \item $M$ is a matching allowing parallel edges: Any two edges in $M$ are either completely disjoint or parallel copies of each other.
        \item There exists a good cut $(X,Y)$ with $|C(A,B) \Delta C(X,Y)| \le k'$ and $M \subseteq E(X, Y) \cap E(A, B)$.
    \end{itemize}

    \noindent\textbf{Requirement:} For \textit{each} function $f : T \to \{X, Y\}$ and $0 \le k'' \le k$, output a cut $(X_{f,k''}, Y_{f,k''})$ satisfying all the following or output nothing:
    \begin{enumerate}
        \item {consistent with $f$}: $f(v) = X$ iff $v \in X_{f, k''}$ for all $v \in T$,
        \item {consistent with $M$}: $M \subseteq E(A, B) \cap E(X_{f,k''}, Y_{f,k''})$,
        \item {$k''$-close to $(A, B)$}: $|C(A, B) \Delta C(X_{f,k''}, Y_{f,k''})| \leq k''$,
    \end{enumerate}
    For any $(X,Y)$ which is a good cut that satisfies $|C(A,B) \Delta C(X,Y)| \le k''$ and $M \subseteq E(X, Y) \cap E(A, B)$ for some $k'' \le k'$, the output $(X_{f^*,k''}, Y_{f^*,k''})$, where $f^\star$ is consistent to $(X,Y)$, should additionally be a good cut (not necessarily equal to $(X, Y)$).
\end{mdframed}

Clearly, \CUT is a special case with $M = T = \varnothing$ after doing the pre-processing in \Cref{lem:2ae-edge2partition}. In the following we introduce the algorithm solving \CUTTERMINAL.

Let $q := k^2  2^{2k+2}$. Before introducing the algorithm, we require one additional definition.
\begin{definition}[\textnormal{$(k,q)$-cut}]\textnormal{
    Let $G=(V,E)$ be a multigraph and let $M\subseteq E$ be the set of marked edges which is a matching allowing parallel edges. A cut $L \subseteq V$ is called $(k, q)$-balanced (more simply a $(k,q)$-cut) if the following conditions are met. 
\begin{itemize}
    \item $|E(L, V \setminus L)| \leq k$.
    \item Both $G(L)$ and $G(V \setminus L)$ are connected (using both edges inside and outside $M$).
    \item Both $G(L)$ and $G(V \setminus L)$ contains at least $q$ non-marked edges (edges outside $M$).
\end{itemize} }
\end{definition}

Our algorithm follows a simple win-win strategy. If there is no $(k, q)$-cut, a simple color-coding algorithm will solve the terminal problem. If there is a $(k, q)$-cut, by recursively solving a smaller subproblem, one can reduce the number of unmarked edges. Combining this with the algorithm to compute a $(k, q)$-cut, one can solve the terminal version in FPT time. 

Now we present each of the three parts. 
We start by directly solving the case where there is no $(k, q)$-cut, using the (derandomized) coloring coding technique.

\begin{claims}
One can solve \CUTTERMINAL\ in $2^{O(k \log q)} n^{O(1)}$ time when there is no $(k,q)$-cut in the instance. 
\label{claim:color}
\end{claims}
\begin{proof}Given an \CUTTERMINAL\ instance $I = $ $(G, k', \mathcal{P}= C(A,B), t, T, M)$, the algorithm considers every $f : T \to \{ X, Y \}$ and $0 \le k'' \le k'$ and does the following. Assume that $f$ and $k''$ is consistent with some good cut $(X, Y)$, because apart from running time, the requirement only concerns for the correct $f$ and $k''$. 

Note that the set of edges $S = C(A, B) \Delta C(X, Y) = E(A,B) \Delta E(X,Y) = E(AX, BX) \cup E(AY, BY) \cup E(AX, AY) \cup E(BX, BY)$ forms a cut between $L := AX \cup BY$ and $R := AY \cup BX$, no matter how $t$ assigns type to edges. Since $G$ is connected and $|S| \leq k''$, $G \setminus S$ has at most $k'' + 1$ connected components.

Assume towards a contradiction that there are two connected components $L', R'$ of $G \setminus S$ such that $L' \subseteq L$ and $R' \subseteq R$ and both $G(L')$ and $G(R')$ have at least $q$ non-marked edges. Then the minimum cut separating $L'$ and $R'$ in $G$ will cut at most $k'' \le k$ edges where the two resulting parts are connected and contain $L'$ and $R'$ respectively, which contradicts the fact that there is no $(k, q)$-cut. \footnote{We use the fact that any minimum $st$-cut in a connected graph have both sides connected.}

Therefore, at least one of $L$ and $R$ (say $R$) has the property that every connected component of $G \setminus S$ contained in $R$ has at most $q$ non-marked edges. Since marked edges form a matching, the number of vertices in such a component can be at most $2q + 2$, so we have $|R| = O(qk)$. 

Let $\{R_i\}_{i=1}^\ell$ be the connected components of $G(R)$. Define\\ $a_i^+ = |C(A \Delta R_i, B \Delta R_i) \backslash C(A,B)|$ and $a_i^- = |C(A,B) \backslash C(A \Delta R_i, B \Delta R_i)|$. Notice that $|C(X,Y)| = |C(A,B)| + \sum_{i=1}^\ell (a_i^+ - a_i^-)$ and $|C(X,Y)\Delta C(A,B)| = \sum_{i=1}^\ell (a_i^+)+(a_i^-) \le k'$. Guessing the correct values of $\ell$ and $a^+_i, a^-_i$ takes at most $k^{O(k)}$ time.

Let $N = N(R)$ be the open neighbor of $R$ in $G$. Note that $|N| \leq |E(R, V \setminus R)| = |C(A, B) \Delta C(X, Y)| \leq k''$. Use \Cref{thm:splitter} with $a = |R|, b = k$ to have a family of at most $2^{O(k \log q)} \log n$ colorings that color each vertex using $\{ 0, 1 \}$. There exists a coloring $\chi$ that colors $R$ with $1$ and colors $N$ with $0$. For each $i \in [\ell]$, find all candidates $R'_i \subseteq V$ that is {\em identical} to $R_i$ with respect to the guessed information about $R_i$. Formally, $R'_i$ is a connected component of $G$ after deleting all vertices of color $0$, and in line with the guessing $a_i^+$ and $a_i^-$.

For the correct coloring $\chi$, such $R'_i$ exists since $R_i$ satisfies the condition. Also, finding all such candidates $R'_i$ is easy in polynomial time since one only needs to scan connected components after deleting color-$0$ vertices.

For any $R'_1, \cdots, R'_{\ell}$ that are disjoint, since $R'_i$ and $N(R'_j)$ are disjoint for any $i, j \in [\ell]$, the set of edges newly cut/uncut by updating $A \leftarrow A \Delta R'_i$ and $B \leftarrow B \Delta R'_i$ are disjoint for all $i$. 
Therefore, letting $R' := \cup_i R'_i$ and letting $A' \leftarrow A \Delta R'$ and $B' \leftarrow B \Delta R'$ the resulting partition has $|C(A', B')| = |C(A, B)| + \sum_i (a^+_i - a^-_i) = |C(X, Y)|$ and $|C(A', B') \Delta C(A, B)| = \sum_i (a^+_i + a^-_i) \le k''$, so the solution is good and $k''-$close to $(A,B)$.

We finally need the solution $(A',B')$ to be consistent with both $f$ and $M$. For any vertex $v \in T$, if $v$ has color $0$, then $v \in A$ if and only if $v \in X$ (recall $R = AY \cup BX$), otherwise by including $v$ in $R$, one can change the belonging of $v$. So the consistency of $f$ requires certain components to be included into or excluded from $R'_i$. For $(u,v) \in M$, if $u$ and $v$ has the same color then $M \in E(A,B) \cap E(X,Y)$ if and only if $M \in E(A,B)$ which is always satisfied, otherwise $M \in E(A,B) \cap E(X,Y)$ if and only if the label-1 vertex is not in $R$, so consistency of $M$ requires certain components to be excluded from $R'_i$.

The consistency constraints are all in the form of ``removing a component from the candidates'' or ``impose certain components to be selected by $R'_i$'', which is easy to find out and handle in polynomial time. 

We finally calculate the running time of the algorithm. The algorithm first tries all possible $(f,k'')$ in $2^{O(k)}$ time, tries all colorings in $2^{O(k \log q)} n^{O(1)}$ time, and finally guesses the correct values of $\ell$ and $a_i^+,a_i^-$ in $2^{O(k \log k)}$ time. All other processes are in polynomial time. So the total running time is $2^{O(k \log q)} n^{O(1)}.$
\end{proof}

Next we show how to compute a $(k, q)$-cut if it exists. 

\begin{claims}
There exists an algorithm that runs in time $2^{O(k \log q)} n^{O(1)}$ and returns a $(k ,q)$-balanced cut for a connected graph if it exists. 
\label{claim:finding-balanced-cut}
\end{claims}
\begin{proof}
Suppose that $L \subseteq V$ is a $(k, q)$-cut and $R = V \setminus L$. We first show that there exists $E_L \subseteq E(L)$ that induces a connected subgraph and has at most $O(q)$ edges but at least $q$ non-marked edges.
Consider a procedure where we start from $V_L = \{ v \}$ for an arbitrary vertex $v \in L$, and 
iteratively add one arbitrary vertex $u \in L \cap N(V_L)$ to $V_L$ until $E(G(V_L))$ contains at least $q$ non-marked edges. Since $E(G(V_L))$ is connected and increased by at least one while marked edges form a matching, $|V_L| \leq 2q + 2$. Then we choose $E_L \subseteq E(G(V_L))$ that contains at least $q$ non-marked edges while $(V_L, E_L)$ is connected and $|E_L| \leq O(q)$. Define $V_R$ and $E_R$ similarly. 

Use~\Cref{thm:splitter} with $a = O(q)$ and $b = k$ to try at most $2^{O(k \log q)} \log n$ colorings of edges with $2$ colors. One coloring $\chi$ colors $E(L, R)$ with $0$ and $E_L \cup E_R$ with $1$. After deleting all edges of color $0$, there are two connected components $V'_L$ and $V'_R$ such that $E_L \subseteq E(G(V'_L))$ and $E_R \subseteq E(G(V'_R))$. Note that $V'_L \neq V'_R$ because $L$ and $R$ are separated by removing edges of color $0$. 

Then, trying every pair $(U', V')$ of connected components after deleting $0$-colored edges and computing the minimum cut separating $U'$ and $V'$ in the original graph will yield a cut of size at most $k$ where each side is connected has at least $q$ non-marked edges. We use the fact that in a connected graph, any minimum $s$-$t$ cut has both of its sides connected. 
\end{proof}

Finally, we prove that the existence of a $(k, q)$-cut, combined with a recursion, allows us to make progress by reducing the number of unmarked edges. 
Let $T(m)$ be the running time of our algorithm to solve the terminal version with $m$ unmarked edges. 

\begin{claims}
Suppose there exists a $(k,q)$-cut, in time $T(\ell) + n^{O(1)}$, one can reduce the current instance to another \CUTTERMINAL\ instance with at least $\ell - q/2$ fewer unmarked edges for some $\ell \in [q, m - q]$. 
\label{claim:recurse}
\end{claims}
\begin{proof}
Let $(L, R)$ be a $(k, q)$-cut. Without loss of generality, suppose that $L$ has at most $k$ terminals from $T$, and let $T_L \subseteq L$ be the set of vertices that have an edge to $R$. Letting $T' = (T \cap L) \cup T_L \subseteq L$ be the new set of terminals for $L$, we know that $|T'| \leq 2k$. Let $\ell$ be the number of unmarked edges in $L$ and $m$ be the total number of unmarked edges. 

Recursively solve the terminal version with input $(G(L), (L \cap A, L \cap B), k', T', M \cap E(G(L)))$. For each $f' : T' \to \{ X, Y \}$ and $0 \le k'' \le k'$, it returns a partition $(X_{f', k''}, Y_{f', k''})$ of $L$ that is (1) consistent with $f'$, (2) consistent with $M \cap E(G(L))$, and (3) $k''$-close to 
$(L \cap A, L \cap B)$ (or nothing).

Let $f^* : T' \to \{ X, Y \}$ be the  correct guessing with respect to the good partition $(X, Y)$ and $k^*$ be the correct closeness parameter $|C(A \cap L, B \cap L) \Delta C(X \cap L, Y \cap L)|$ in $L$. We update $(X, Y)$ with $X \leftarrow (X \cap R) \cup X_{f^*, k^*}$ and $Y \leftarrow (Y \cap R) \cup Y_{f^*, k^*}$. Note that the new $(X, Y)$ is still a good cut since $(X_{f^*, k^*}, Y_{f^*, k^*})$ is good within $L$ and the interaction between $L$ and $R$ only depends on the terminals $T'$, where the previous and new solutions agree. By the guarantee (2) and (3) in the above paragraph, the new $(X, Y)$ is still $k'$-close to $(A, B)$ and consistent with $M$. 

Since $L$ has at least $q=k^2  2^{2k+2}$ non-marked edges and there are at most $2^{2k} (k+1)$ possible values for $f'$ and $k''$, we have $|E(L) \setminus (\cup_{f', k''} (C(A \cap L, B \cap L) \Delta C(X_{f', k''},Y_{f', k''}))| \geq \ell - q/2$. These edges do not belong to $C(A, B) \Delta C(X, Y) = E(A,B) \Delta E(X,Y)$ for sure. For each such edge $e$, perform the following operation. 

\begin{itemize}
    \item If $e \notin E(A, B)$, which means that $e$ is not cut in the current solution, this means that $e$ is not cut in the good solution too. Contract it. 
    
    \item If $e \in E(A, B)$, mark it. To ensure that the set of marked edges $M$ is a matching, if $e = (u, v), f = (v, w)$ are both in $M$ with $u \neq w$, then we know that $u, w$ will be in the same side in the good solution, so we can contract them.
\end{itemize}

By doing this, the number of unmarked edges is decreased by at least $\ell - q/2$. 
\end{proof}

With the three main claims, we finish the proof of~\Cref{thm:main-cut}. 

\begin{proof} [Proof of~\Cref{thm:main-cut}]
From the definition of the terminal version, given graph $G$ and the current solution $(A, B)$, running the terminal version with $T = M = \emptyset$ and $k' = k$ we can find a good cut given the promise $|C(A, B) \Delta C(X, Y)| \leq k$ for some (possibly different) good cut $(X, Y)$. 

Let us analyze the running time. As defined previously, let $T(m)$ be the running time of our algorithm with $m$ unmarked edges. Since the graph is connected and marked edges form a matching, we have $m \ge \frac{n}{2}$. Our algorithm first uses \Cref{claim:finding-balanced-cut} to find a $(k, q)$-cut in time $2^{O(k \log q)} n^{O(1)}$. If there is none, it uses \Cref{claim:color} to solve the terminal version in time $2^{O(k \log q)} n^{O(1)}$. Otherwise, it uses \Cref{claim:recurse} to reduce the number of unmarked edges by $\ell - q/2$ in time $T(\ell) + n^{O(1)}$, for some $\ell \in [q, m - q]$. Therefore, we have the following recurrence relation: 

\[
T(m) \leq 2^{O(k \log q)} n^{O(1)} + \max_{\ell \in [q, m - q]} \bigg( T(\ell) + T(m - (\ell - q/2)) + n^{O(1)} \bigg).
\]

Using $m \in \left[\frac{n}{2}, n^2\right]$, $T(m) \leq 2^{O(k \log q)} n^{O(1)}$ satisfies the above.
\end{proof}

\section{Hardness Proof for $r$AE for $r \ge 3$}\label{sec:ae}

In this section, we prove \Cref{characterization:3ae}. We start with a simple lemma addressing that hardness result for 3AE is enough.

\begin{lemma}\label{lem:3tor}
    If \MinCSPwP$(3\ALLEQ)$ is W[1]-hard, then for any integer $r \ge 3$, \MinCSPwP$(r\ALLEQ)$ is W[1]-hard.
\end{lemma}
\begin{proof}
    For each integer $r \ge 3$, we give a reduction from \MinCSPwP$(r\ALLEQ)$ to \MinCSPwP$((r+1)\ALLEQ)$, and the statement follows from induction. Given an \MinCSPwP$(r\ALLEQ)$ instance $(I,k,\mathcal{P})$, the \MinCSPwP$((r+1)\ALLEQ)$ instance $(I',k,\mathcal{P}')$ is constructed in the following way: \begin{itemize}
        \item Initially, $V' = V$.
        \item For each clause $C \in \mathcal{C}$, introduce a variable $v_C$ to $V'$. Choose any variable $v$ whose literal appears in $C$, and introduce a clause $C'$ to $\mathcal{C}'$, which is satisfied when $C$ is satisfied and additionally $v = v_C$. It is easy to see that $C'$ is a $(r+1)\ALLEQ$ clause. If $C \in \mathcal{P}$ then add $\mathcal{C}'$ to $\mathcal{P}'$.
    \end{itemize}

    Since variable $v_C$ is only incident to $C'$, $v = v_C$ in the good assignment for any clause $C$. Then the equivalence of two instances is straightforward.
\end{proof}

Using this approach of adding dummy variables, it is straightforward to further show the following.
\begin{claims}\label{claim:3aeto123ae}
    For all integer $r \ge 3$, if \MinCSPwP$(\cup_{r' \le r} r'\ALLEQ)$ is W[1]-hard, then \MinCSPwP$(r\ALLEQ)$ is W[1]-hard.
\end{claims}

Now we introduce the auxiliary W[1]-hard problem we will use in our hardness reduction, which originates from \cite{marx2009constant} and serves as the source problem of all hardness reductions in \cite{kim2023flow}.

\begin{framed}
\noindent\textbf{Problem:} $\PMC$.

\noindent\textbf{Input:} A directed acyclic graph $G = (V, E)$, vertices $s, t \in V$, an integer $l \in \mathbb{N}$, a partition $\mathcal{E}$ of $E$ into pairs (sets of size $2$), and a partition of $E$ into a set $\mathcal{F}$ of $2l$ arc-disjoint $st$-paths.

\noindent\textbf{Parameter:} $l$.

\noindent\textbf{Output:} An $st$-Cut $Z$ that touches (has non-empty intersection with) %
at most $l$ pairs from $\mathcal{E}$ if exists, otherwise report that such cut does not exist.
\end{framed}

\begin{lemma}[\cite{kim2023flow}]\label{lem:pmchardness}
$\PMC$ is W[1]-hard.
\end{lemma}

We first present the proof for \MinCSPwP$(\cup_{r' \le 4} r'\ALLEQ)$, which is more intuitive.

\begin{theorem}\label{thm:4ae-hardness}
    \MinCSPwP$(\cup_{r' \le 4} r'\ALLEQ)$ is W[1]-hard.
\end{theorem}
\begin{proof}
    We give a reduction from $\PMC$, whose hardness is by \Cref{lem:pmchardness}, to \MinCSPwP$(\cup_{r' \le 4} r'\ALLEQ)$. Given an $\PMC$ instance $(G,s,t,l,\mathcal{E},\mathcal{F})$, we construct the following \MinCSPwP$(\cup_{r' \le 4} r'\ALLEQ)$ instance $(I, k, \mathcal{P})$: 
    \begin{itemize}
        \item For each vertex $u \in V_G$, introduce a variable $x_u$ in $V_I$.
        \item For each edge $(u, v) \in E$, introduce a $2\ALLEQ$ clause $x_u = x_v$ in both $\mathcal{C}_I$ and $\mathcal{P}$. In the rest of this proof, we refer to such clauses as \emph{edge clauses}.%
        \item For each pair $\left((u_1,v_1), (u_2,v_2)\right) \in \mathcal{E}$, introduce a $4\ALLEQ$ clause $x_{u_1} = \overline{x_{v_1}} = x_{u_2} = \overline{x_{v_2}}$ in $\mathcal{C}_I$.
        \item Further introduce $(l+1)$ copies of the $2\ALLEQ$ clause $x_s \ne x_t$ to $\mathcal{C}_I$. Set $k = 4l+1$.
    \end{itemize}

    Now we prove the equivalence of the instances. Suppose the given $\PMC$ instance has an $st$-cut $Z$ touching $l$ pairs, we construct the following assignment $\alpha$ for the \MinCSPwP$(\cup_{r' \le 4} r'\ALLEQ)$ instance: Assign $x_u$ to $0$ for $u \in V_G$ if $u$ is in the $s$-side of $Z$, otherwise, when $u$ is in the $t$-side, assign $x_u$ to $1$. $\alpha$ satisfies the following clauses: \begin{itemize}
        \item All edge clauses except for the $2l$ clauses corresponding to edges in $Z$;
        \item $l$ $4\ALLEQ$ clauses corresponding to the $l$ edge pairs in $Z$ (Notice that the tails of all edges in $Z$ must both belong to $s$-side);
        \item $(l+1)$ copies of the $x_s \ne x_t$ clause.
    \end{itemize} So $|C_I(\alpha) \Delta \mathcal{P}|\le 4l+1$ and $|C_I(\alpha)| - |\mathcal{P}| = 1$. Now we only need to check $\alpha$ is good.
    
    Consider any assignment $\beta: V_I \to \{0,1\}$ with $|C_I(\beta)| > |\mathcal{P}|$. Define $Z_\beta = \{(u,v) \in E \mid x_u \ne x_v\}$ and $\mathcal{E}_{\beta} = \{(e_1,e_2) \in \mathcal{E} \mid e_1 \in Z_\beta, e_2 \in Z_\beta\}.$ We have the following:
    \begin{itemize}
        \item $|C_I(\beta)| \le |E| - |Z_\beta| + |\mathcal{E}_\beta| + (l+1)[x_s \ne x_t]$ (It is not equality since pairs in $\mathcal{E}_\beta$ may not satisfy the corresponding $4\ALLEQ$ clauses). Since we further have $|E| = |\mathcal{P}| < |C_I(\beta)|$, we have $|\mathcal{E}_\beta| + (l+1)[x_s \ne x_t] > |Z_\beta|$.
        \item $0 \le |\mathcal{E}_\beta| \le \lfloor \frac{|Z_\beta|}{2} \rfloor$, since $\mathcal{E}$ is a partition of $E$. Hence, $x_s \ne x_t$ must hold for such $\beta$. Further, we have $l+1 > |Z_\beta| - |\mathcal{E}_\beta| \geq \lceil\frac{|Z_\beta|}{2}\rceil$, so $|Z_\beta| \le 2l$.
        \item $x_s \ne x_t \Rightarrow |Z_\beta| \ge 2l$, since $Z_\beta$ corresponds to an $st$-cut and $\mathcal{F}$ implies that the minimum $st$-cut have size $2l$. Therefore, $\beta$ must satisfy $|Z_\beta| = 2l, |\mathcal{E}_\beta| = l$ and $x_s \ne x_t$, which corresponds to a solution for the $\PMC$ instance. Further, in this instance, the tails of the $2l$ edges in $Z_\beta$ must belong to $s$-side of $Z_\beta$, so that all the clauses corresponding to $\mathcal{E}_\beta$ are satisfied, and we have $|C_I(\beta)| - |\mathcal{P}| = 1$.
    \end{itemize}
    So any assignment $\beta$ with $|C_I(\beta)| > |\mathcal{P}|$, if it exists, corresponds to a feasible solution to the given $\PMC$ instance, and is a good assignment for the constructed \MinCSPwP$(\cup_{r' \le 4} r'\ALLEQ)$ instance. This shows the validity (satisfying the promise) of the \MinCSPwP$(\cup_{r' \le 4} r'\ALLEQ)$ instance.

    Finally, if the given $\PMC$ instance does not have a $st$-cut $Z$ touching $l$ pairs, following the previous argument, we cannot find $\beta$ with $|C_I(\beta)| > |\mathcal{P}|$. Notice that in this case, %
    if we set all variables to zero, the assignment exactly satisfies $\mathcal{P}$, which is a good solution now, so the instance is still valid.

    In summary, by distinguishing the number of satisfied clauses in the good assignment and $|\mathcal{P}|$, we can determine whether an $st$-cut $Z$ in the $\PMC$ instance satisfying the condition exists. If it exists, we can turn the good assignment of the \MinCSPwP$(\cup_{r' \le 4} r'\ALLEQ)$ instance to an $st$-cut satisfying the conditions for the $\PMC$ problem.
\end{proof}

To extend the proof in \Cref{thm:4ae-hardness} to $3\ALLEQ$, we need to mimic the $4\ALLEQ$ clause $x_{u_1} = \overline{x_{v_1}} = x_{u_2} = \overline{x_{v_2}}$ corresponding to edge pairs using $3\ALLEQ$ clauses. The most straightforward way with some clause duplication suffices to prove the main hardness result.

\begin{theorem}\label{thm:3ae-hardness}
    \MinCSPwP$(\cup_{r' \le 3} r'\ALLEQ)$ is W[1]-hard.
\end{theorem}
\begin{proof}
    We give a reduction from $\PMC$, whose hardness is by \Cref{lem:pmchardness}, to \MinCSPwP$(\cup_{r' \le 3} r'\ALLEQ)$. Given an $\PMC$ instance $(G,s,t,l,\mathcal{E},\mathcal{F})$, we construct a \MinCSPwP$(\cup_{r' \le 3} r'\ALLEQ)$ instance $(I, k, \mathcal{P})$ in the following way: 
    \begin{itemize}
        \item For each vertex $u \in V_G$ introduce a variable $x_u$ in $V_I$.
        \item For each edge $(u, v) \in E$ introduce $5$ copies of the $2\ALLEQ$ clause $x_u = x_v$ in both $\mathcal{C}_I$ and $\mathcal{P}$.  In the rest of this proof, we refer to such clauses as \emph{edge clauses}.%
        \item For each pair $\left((u_1,v_1), (u_2,v_2)\right) \in \mathcal{E}$ introduce $2$ copies of each of the following $4$ $3\ALLEQ$ clauses in $\mathcal{C}_I$: $x_{u_1} = x_{u_2} = \overline{x_{v_2}}, \overline{x_{v_1}} = x_{u_2} = \overline{x_{v_2}}, x_{u_1} = \overline{x_{v_1}} = x_{u_2}$ and $x_{u_1} = \overline{x_{v_1}} = \overline{x_{v_2}}$. These are all four 3AE clauses generated from removing one variable in the 4AE clause $x_{u_1} = x_{u_2} = \overline{x_{v_1}} = \overline{x_{v_2}}$.
        \item Further introduce $(2l+1)$ copies of the $2\ALLEQ$ clause $x_s \ne x_t$ to $\mathcal{C}_I$. Set $k = 20l+1$.
    \end{itemize}

    Now we prove the equivalence of the instances. Suppose the given $\PMC$ instance has an $st$-cut $Z$ touching $l$ pairs, we construct exactly the same assignment $\alpha$ for the \MinCSPwP$(\cup_{r' \le 3} r'\ALLEQ)$ instance as we do in \Cref{thm:4ae-hardness}. $\alpha$ satisfies the following clauses: \begin{itemize}
        \item All edge clauses except for $10l$ clauses corresponding to edges in $Z$;
        \item $8l$ $3\ALLEQ$ clauses corresponding to the $l$ pairs in $Z$;
        \item $(2l+1)$ copies of $x_s \ne x_t$ clause.
    \end{itemize}So $|C_I(\alpha) \Delta \mathcal{P}|\le 20l+1$ and $|C_I(\alpha)| - |\mathcal{P}| = 1$. Now we only need to check $\alpha$ is good.
    
    Consider any assignment $\beta: V_I \to \{0,1\}$ with $|C_I(\beta)| > |\mathcal{P}|$. Define $Z_\beta = \{(u,v) \in E \mid x_u \ne x_v\}, \mathcal{E}_{\beta} = \{(e_1,e_2) \in \mathcal{E} \mid e_1 \in Z_\beta, e_2 \in Z_\beta\}$ and $\mathcal{E}'_{\beta} = \{(e_1,e_2) \in \mathcal{E} \mid [e_1 \in Z_\beta] + [e_2 \in Z_\beta] = 1\}$. We have the following:
    \begin{itemize}
        \item $|C_I(\beta)| \le 5|E| - 5|Z_\beta| + 8|\mathcal{E}_\beta| + 2|\mathcal{E}'_\beta| + (2l+1)[x_s \ne x_t]$, 
        since one can easily check that exactly one of the 4 $3\ALLEQ$ clauses are satisfied for all pairs in $\mathcal{E}_\beta'$, and pairs in $\mathcal{E}_\beta$ can satisfy at most 4 $3\ALLEQ$ clauses. Since we further have $5|E| = |\mathcal{P}| < |C_I(\beta)|$, $\beta$ satisfies $8|\mathcal{E}_\beta| + 2|\mathcal{E}'_\beta| + (2l+1)[x_s \ne x_t] > 5|Z_\beta|$.
        \item $|Z_\beta| = 2|\mathcal{E}_\beta| + |\mathcal{E}'_\beta|$, since $\mathcal{E}$ is a partition of $E$. Hence, $x_s \ne x_t$ must hold for such $\beta$. Further, we have $2l+1 > 2|\mathcal{E}_\beta| + 3|\mathcal{E}'_\beta|$, so $|Z_\beta| = 2|\mathcal{E}_\beta| + |\mathcal{E}'_\beta| \le 2|\mathcal{E}_\beta| + 3|\mathcal{E}'_\beta| < 2l+1$.
        \item $x_s \ne x_t \Rightarrow |Z_\beta| \ge 2l$, since $Z_\beta$ corresponds to an $st$-cut and $\mathcal{F}$ implies that the minimum $st$-cut have size $2l$. Therefore, $|Z_\beta| = 2l$ and $|\mathcal{E}'_\beta| = 0$, which means $|\mathcal{E}_\beta| = l$. Along with $x_s \ne x_t$, the assignment corresponds to a feasible solution for the $\PMC$ instance.  Further, in such instance, the tails of the $2l$ edges in $Z_\beta$ must belong to $s$-side of $Z_\beta$, so all $8l$ clauses corresponding to $\mathcal{E}_\beta$ are satisfied, and we have $|C_I(\beta)| - |\mathcal{P}| = 1$.
    \end{itemize}
    So any assignment $\beta$ with $|C_I(\beta)| > |\mathcal{P}|$, if it exists, corresponds to a feasible solution to the given $\PMC$ instance, and is a good assignment for the constructed \MinCSPwP$(\cup_{r' \le 3} r'\ALLEQ)$ instance. This shows the validity (satisfying the promise) of the \MinCSPwP$(\cup_{r' \le 3} r'\ALLEQ)$ instance.

    Finally, if the given $\PMC$ instance does not have a $st$-cut $Z$ touching $l$ pairs, following the previous argument, we cannot find $\beta$ with $|C_I(\beta)| > |\mathcal{P}|$. Notice that in this case, if we set all variables to zero, the assignment exactly satisfies $\mathcal{P}$, which is a good solution now, so the \MinCSPwP$(\cup_{r' \le 3} r'\ALLEQ)$ is still valid.
    
    In summary, by distinguishing the number of satisfied clauses in the good  assignment and $|\mathcal{P}|$, we can determine whether an $st$-cut $Z$ in the $\PMC$ instance satisfying the condition exists. If it exists, we can turn the good assignment of the \MinCSPwP$(\cup_{r' \le 3} r'\ALLEQ)$ instance to an $st$-cut satisfying the conditions for the $\PMC$ problem.
\end{proof}

\begin{proof}[Proof of~\Cref{characterization:3ae}]
    From \Cref{thm:3ae-hardness}, \MinCSPwP$(\cup_{r' \le 3} r'\ALLEQ)$ is W[1]-hard. Along with \Cref{claim:3aeto123ae}, \MinCSPwP$(3\ALLEQ)$ is W[1]-hard, and the statement follows from \Cref{lem:3tor}.
\end{proof}

\section{Hardness Proof for $\LE$}\label{sec:le}

In this section we prove the W[1]-hardness of \textsc{ImproveMaxCSP}$(\LE)$ for all $r \ge 2$. We first prove the case for $r = 2$, which is \textsc{ImproveMaxCSP}$(2\SAT)$, from the following problem called \MIS.

\begin{framed}
\textbf{Problem:} \MIS.

\textbf{Input:} An undirected graph $G = (V, E)$, an integer $l \in \mathbb{N}^+$, and a partition $(V_1,\cdots,V_l)$ of $V$ into $l$ non-empty sets.

\textbf{Parameter:} $l$.

\textbf{Output:} An independent set of $G$ that contains exactly one vertex from each $V_i$ if exists, otherwise report that such an independent set does not exist.
\end{framed}

The reader can find the W[1]-hardness of \MIS in Corollary 13.8 of \cite{cygan2015parameterized}.

\begin{lemma}[\cite{cygan2015parameterized}]\label{lemma:mishardness}
    \MIS is W[1]-hard.
\end{lemma}

\begin{theorem}\label{thm:2sat-hardness}
    \textsc{ImproveMaxCSP}(2SAT) is W[1]-hard.
\end{theorem}

\begin{proof}
    We give a reduction from \MIS, whose hardness is given by \Cref{lemma:mishardness}, to \textsc{ImproveMaxCSP}$(2\SAT)$. Suppose that $G(V_1,\cdots,V_l; E)$ is the given \MIS instance. We assume there is no isolated vertex in $G$, otherwise we can simply include them in the independent set. We construct the following \textsc{ImproveMaxCSP}$(2\SAT)$ instance $(I,k,\mathcal{P})$:

    For each vertex $u \in V$, introduce a variable $x_u$ in $I$ and a clause $\overline{x_u}$ in $\mathcal{C}$. For each edge $(u,v) \in E$, add $2$ copies of the clause $x_u \lor x_v$ in $\mathcal{C}$ and $\mathcal{P}$. For each $i \in [l]$ and $u \neq v \in V_i$, add $2$ copies of the clause $x_u \lor x_v$ in $\mathcal{C}$ and $\mathcal{P}$. Set $k$ to be $l$.

    We first show that, any optimal assignment satisfies $\mathcal{P}$. Assume towards contradiction that some $\alpha$ with the smallest cost have $x_u = x_v = 0$ for some clause $x_u \lor x_v$ in $\mathcal{P}$. By flipping $\alpha(x_u)$, at least two clauses change from unsatisfied to satisfied (since we introduce 2 copies for each binary clause, and $x_u \lor x_v$ becomes satisfied), and one unary clause changes from satisfied to unsatisfied.
    So the new assignment would satisfy strictly more clauses than $\alpha$, contradicting with the selection of $\alpha$.

    Define $V_\alpha = \{u \mid x_u = 0\}$. Since the optimal $\alpha$ satisfies $\mathcal{P}$, from the construction of $\mathcal{P}$, $V_\alpha$ is a multicolored independent set, whose size is bounded by $l$. Since $|\mathcal{C}_I(\alpha) \Delta \mathcal{P}| = |\mathcal{C}_I(\alpha) \backslash \mathcal{P}| = |V_\alpha|$, the \textsc{ImproveMaxCSP}$(2\SAT)$ instance satisfies the promise, and the maximum multicolored independent set in $G$ is exactly $V_\alpha$ for the optimal assignment $\alpha$ in the \textsc{ImproveMaxCSP}$(2\SAT)$ instance. This finishes the proof.
\end{proof}

Finally, by introducing dummy variables, one can easily reduce $2\SAT$ to $\LE$, showing the main hardness result in this section.

\begin{proof}[Proof of~\Cref{characterization:le}]
    We give a reduction from \textsc{ImproveMaxCSP}$(2\SAT)$, whose hardness was proved in \Cref{thm:2sat-hardness}, to \textsc{ImproveMaxCSP}$(\LE)$ for $r \ge 3$. Given an instance $(I,k, \mathcal{P})$ of \textsc{ImproveMaxCSP}$(2\SAT)$, we would construct the instance $(I', k, \mathcal{P}')$ as follows: \begin{itemize}
        \item $V' = V \cup \{u_i \mid 1 \le i \le r-2\}$.
        \item For each binary clause $\mathcal{C} = (v_a \oplus p_{v_a}) \lor (v_b \oplus p_{v_b}) \in \mathcal{C}$ for $v_a \ne v_b \in V$ and $p_{v_a}, p_{v_b} \in \{0,1\}$, introduce a clause $\land_{x \ne y \in \{u_i \mid 1 \le i \le r-2\} \cup \{v_a,v_b\}} (x \oplus p_x) \lor (y \oplus p_y)$ in $\mathcal{C}'$ (notice this clause is exactly $\LE$ clause), where $p_{u_i} = 0$ for all $i \in [r-2]$. If $C \in \mathcal{P}$, add this clause to $\mathcal{P}'$. For unary clauses process similarly.
    \end{itemize}
    Since all $u_i$'s appear positively in $I'$, any good solution for $(I', k, \mathcal{P}')$ always sets $u_i$ to $1$. By assigning $u_i$'s to $1$ in instance $(I', k, \mathcal{P}')$ (see \Cref{def:assignvalue}), we get exactly the instance $(I,k ,\mathcal{P})$. This proves the equivalence of the two instances.
\end{proof}

\section{Hardness from MinCSPs}\label{sec:hardness-from-fpt}

In this section, we prove \Cref{characterization:main-hardfrommincsp}.

\subsection{Structural Characterization of Boolean Constraint Languages}

We give the following definitions characterizing boolean constraint languages to introduce \Cref{thm:fptdichotomy}. All of these definitions are as in \cite{kim2023flow}.

A boolean boolean relation $R$ is \emph{bijunctive} if it is expressible as a conjunction of 1- and 2-clauses (allowing negations), 
and $R$ is \emph{IHS-B-} (respectively \emph{IHS-B+}) if it is expressible as a conjunction of negative clauses $(\neg x_1 \lor \ldots \lor \neg x_r)$, positive 1-clauses $(x)$, and implications ($x \to y$, negation not allowed) (respectively positive clauses, negative 1-clauses, and implications). Here, IHS-B is an abbreviation for \emph{implicative hitting set, bounded}. A boolean constraint language $\Gamma$ is bijunctive if every boolean relation $R \in \Gamma$ is bijunctive, and is IHS-B+ (respectively IHS-B-) if every relation in $\Gamma$ is IHS-B+
(respectively IHS-B-). Finally, $\Gamma$ is IHS-B if it is either IHS-B+ or IHS-B-. (Note that this is distinct from every boolean relation $R \in \Gamma$ being either IHS-B+ or IHS-B-.)

We say that an undirected graph is \emph{$2K_2$-free} if it does not contain $2K_2$ as an induced subgraph. An directed graph is \emph{$2K_2$-free} if its underlying undirected graph is $2K_2$-free.

\begin{definition}[Gaifman Graph]
 The \emph{Gaifman graph} of a boolean relation $R$ of arity $r$ is the undirected graph $G_R$ with vertex set $V(G_R)=[r]$ and an edge $(i,j)$ for all $1\leq i<j\leq r$ for which there exist $a_i,a_j\in\{0,1\}$ such that $R$ contains no tuple $(t_1,\ldots,t_r)$ with $t_i=a_i$ and $t_j=a_j$. (In other words, there is an edge $(i,j)$ in $G_R$ if and only if constraint $R(x_1,\cdots,x_r)$ has no satisfying assignment $\alpha $ with $\alpha(x_i)=a_i$ and $\alpha(x_j)=a_j$ for some $a_i,a_j\in\{0,1\}$.)
\end{definition}

\begin{definition}[Arrow Graph]
 The \emph{arrow graph} of a boolean boolean relation $R$ of arity $r$ is the directed graph $H_R$ with vertex set $V(H_R)=[r]$ and an edge $i \to j$ for all $1\leq i,j\leq r$ with $i\neq j$ such that $R$ contains no tuple $(t_1,\ldots,t_r)$ with $t_i=1$ and $t_j=0$ but does contain two tuples $t'$ and $t''$ with $t'_i=t'_j=0$ and $t''_i=t''_j=1$. 
 (In other words, there is an edge $(i,j)$ in $H_R$ if and only if the constraint $R(x_1,\cdots,x_r)$ has no satisfying assignment $\alpha$ with $\alpha(x_i)=1$ and $\alpha(x_j)=0$ but has two satisfying assignments $\alpha_1$ and $\alpha_2$ with $\alpha_1(x_i)=\alpha_1(x_j)=0$ and $\alpha_2(x_i)=\alpha_2(x_j)=1$.)
\end{definition}

The following claims are straightforward from the definition and the symmetry of $\SymRel(r,S)$ relations.

\begin{claims}\label{claim:addbijunc}
    If a boolean relation $R$ of arity $r$ is bijunctive, then $R \oplus b$ is bijunctive for any $b \in \mathbb{F}_2^r$.
\end{claims}
\begin{claims}\label{claim:isomorphism}
    For any boolean relation $R$ of arity $r$ and $b \in \mathbb{F}_2^r$, $G_R \cong G_{R \oplus b}$ and $H_R \cong H_{R \oplus b}$, where $\cong$ stands for graph isomorphism relation.
\end{claims}
\begin{claims}\label{claim:graphcomplete}
    For any $r$ and $S \subseteq \{0,1,\cdots,r\}$, $G_{\SymRel(r,S)}$ and the underlying undirected graph of $H_{\SymRel(r,S)}$ is either the complete graph or the empty graph.
\end{claims}
Combining \Cref{claim:isomorphism} and \Cref{claim:graphcomplete}, we immediately get the following corollary.
\begin{corollary}\label{lem:2k2free}
    For any integer $r \in \mathbb{N}^+$, $S \subseteq\{0,1, \cdots, r\}$ and $b \in \mathbb{F}_2^r$, $G_{\SymRel(r,S)\oplus b}$ and the underlying undirected graph of $H_{\SymRel(r,S)\oplus b}$ is either the complete graph or the empty graph.
\end{corollary}

\subsection{Main Theorem on MinCSP Hardness}

We start with a simple claim that allows us to use hardness results of $\MinCSP$ to prove most hardness results for \MinCSPwP.

\begin{claims}\label{clm:prediction2fpt}
    If $\MinCSP(\Gamma)$ is W[1]-hard, then \MinCSPwP$(\Gamma)$ is W[1]-hard.
\end{claims}

\begin{proof}
    We present a reduction from $\MinCSP(\Gamma)$ to \MinCSPwP$(\Gamma)$. For a given $\MinCSP(\Gamma)$ instance $(I,k)$, we simply define $\mathcal{P} := \mathcal{C}_I$, and this would result in a \MinCSPwP$(\Gamma)$ instance $(I,k,\mathcal{P})$. Notice that the instance $(I,k,\mathcal{P})$ may not satisfy the promise, and remind that we specify in \Cref{remarkalgo-boundary} the definition of the \text{ImproveMaxCSP} problem that some assignment should be provided for such instances.

    Since $\cost(\alpha) = |\mathcal{C}_I(\alpha) \Delta \mathcal{C}_I| = |\mathcal{C}_I(\alpha) \Delta \mathcal{P}|$, when $(I,k)$ has an assignment of cost $\le k$, $(I,k,\mathcal{P})$ satisfies the promise and solving $(I,k,\mathcal{P})$ would result in an optimal assignment. Otherwise, solving $(I,k,\mathcal{P})$ would result in an arbitrary assignment whose cost is larger than $k$, and we can conclude by checking its cost that $I$ has no assignment of cost $\le k$.

\end{proof}

According to \cite{kim2023flow}, the complexity of $\MinCSP$ for boolean constraint languages can be fully characterized.

\begin{lemma}[\cite{kim2023flow}]\label{thm:fptdichotomy}
    For a finite boolean constraint language $\Gamma$, $\MinCSP(\Gamma)$ is FPT if it satisfies one of the following conditions, and is W[1]-hard if none of the following holds. \begin{itemize}
        \item For all $R \in \Gamma$, $R$ is bijunctive and $G_R$ is $2K_2$-free.
        \item For all $R \in \Gamma$, $R$ is IHS-B and $H_R$ is $2K_2$-free.
    \end{itemize}
\end{lemma}

\subsection{Case Distinction for Our Setting}

Since complete and empty graphs are all $2K_2$-free, from \Cref{lem:2k2free}, $2K_2$-freeness is always satisfied for $\SymLang(r,S)$. So we only need to check the bijunctive and IHS-B properties for them. Recall two constraint languages are equivalent if they contain the same set of relations.

\begin{theorem}\label{thm:ihsb}
    When $\Gamma = \SymLang(r,S)$ is nontrivial and IHS-B, $\Gamma$ is equivalent to $r\AND$.
\end{theorem}
\begin{proof}

    We show here proof for IHS-B- case, IHS-B+ case is similar. Since $\Gamma$ is nontrivial and allows all negation patterns, there exists a relation $R \in \Gamma$ with IHS-B- definition $\psi$, say with scope $x_1,\cdots,x_r$, rejecting assignment $0^r$. Since $0^r$ satisfies all negative clauses and all implications, there must be a positive 1-clause $x_i$ in $\psi$ for some variable $x_i$. Therefore, all elements $\alpha \in R$ have $\alpha_i = 1$. %

    Now choose any symmetric relation $R_0 \in \Gamma$. %
    Such $R_0$ exists from the definition of $\Gamma$. Since $R_0 = R \oplus b$ for some $b \in \mathbb{F}_2^r$, all elements $\alpha_0 \in R_0$ have the same value in the $i$-th index. Without loss of generality suppose $\alpha_{0,i} = 0$. From the  symmetry of $R_0$, we have $\alpha_{0,j} = 0$ for any $\alpha_0 \in R_0$ and $j \in [r]$. Therefore, $R_0 = \{0^r\} = \SymRel(r,\{0\})$ and $\Gamma$ is $r\AND$.%
\end{proof}

\begin{theorem}\label{thm:bijunctive}
    When $\Gamma = \SymLang(r,S)$ is nontrivial and bijunctive, $\Gamma$ is either $\LE$, $r\ALLEQ$ or $r\AND$.
\end{theorem}
\begin{proof}
    When $r = 1$, $\Gamma$ is clearly bijunctive, 
    and the only nontrivial one is $1$AND. So in the following we consider $r \ge 2$.

    Choose any symmetric relation $R \in \Gamma$. Such $R$ exists from the definition of $\Gamma$. For a term $\psi$ with scope $X$, let us say that $R$ implies $\psi$ if all elements in $R$ satisfies $\psi$. Define $\pi_{ij}(R) = \pi_{\{i,j\}}(R) = \{\alpha_{\{i,j\}} \mid \alpha \in R\}$ for $i,j \in [r]$. %
    By symmetry, $\pi_{ij}(R)$ are all the same for any $i \ne j$. %
    Further, since $R$ is bijunctive, $\land_{ij} \pi_{ij}(R)$ must defines $R$. So $R$ is exactly determined by $\pi_{12}(R)$.
    
    We have $\pi_{12}(R) = \pi_{21}(R)$, so $\pi_{12}(R)$ is symmetric. Write $\pi_{12}(R) = \SymRel(2, S')$. Let us enumerate the set of all possible $S'$: \begin{itemize}
        \item $\{0,1,2\}$: $R = \{0,1\}^r$ which is trivial.
        \item $\{0,1\}$: $R = \SymRel(r, \{0,1\})$ whose corresponding language is $\LE$.
        \item $\{1,2\}$: $R = \SymRel(r, \{r-1, r\})$ whose corresponding language is $\LE$.
        \item $\{0,2\}$: $R = \SymRel(r, \{0,r\})$ whose corresponding language is $r\ALLEQ$.
        \item $\{0\}$: $R = \SymRel(r, \{0\})$ whose corresponding language is $r\AND$.
        \item $\{1\}$: No such $R$ exists for $r \ge 3$. For $r=2, R = \SymRel(2,\{1\})$ corresponds to $2\ALLEQ$.
        \item $\{2\}$: $R = \SymRel(r, \{r\})$ whose corresponding language is $r\AND$.
        \item $\varnothing$: $R = \varnothing$ which is trivial.
    \end{itemize}
    It is easy to check  $\LE$, $r\ALLEQ$ and $r\AND$ are bijunctive.
\end{proof}

Now we prove \Cref{characterization:main-hardfrommincsp}.
\begin{proof} [Proof of~\Cref{characterization:main-hardfrommincsp}]
    For nontrivial $\Gamma = \SymLang(r,S)$ for some $r \in \mathbb{N}^+$ and $S \subseteq \{0,1,\cdots,r\}$ that is not equivalent to $r\AND, r\ALLEQ$ or $\LE$, from \Cref{thm:ihsb}, $\Gamma$ is not IHS-B; from \Cref{thm:bijunctive}, $\Gamma$ is not bijunctive. So from \Cref{thm:fptdichotomy}, $\MinCSP(\Gamma)$ is W[1]-hard. Finally, according to \Cref{clm:prediction2fpt}, \MinCSPwP$(\Gamma)$ is W[1]-hard.
\end{proof}

\bibliography{reference}

\end{document}